\documentclass[letterpaper, 10 pt, conference]{ieeeconf}  
\IEEEoverridecommandlockouts                              
\usepackage{amsmath,amssymb}
\usepackage{amsthm}
\usepackage{xcolor}
\usepackage{tabularx}
\usepackage{booktabs}
\usepackage{multirow}
\usepackage{graphicx}
\usepackage{subcaption}
\usepackage{cite}
\usepackage{flushend}

\title{\LARGE \bf
Output-Feedback Safe Control of  Discrete-Time\\ Stochastic Systems with Chance Constraints
}

\author{Jianing Zhao, Zhuoting Cai and Xiang Yin
\thanks{All the authors are with the School of Automation and Intelligent Sensing, Shanghai Jiao Tong University, Shanghai 200240, China. {{\tt  E-mail: \{jnzhao,marinasjtu,yinxiang\}@sjtu.edu.cn}.}}
}

\def \RR{\mathbb{R}}
\def \NN{\mathbb{N}}

\def \cov{\text{cov}}
\def \PP{\mathbb{P}}

\def \PP{\mathbb{P}}
\def \EE{\mathbb{E}}
\def \u{\mathbf{u}}
\def \y{\mathbf{y}}

\DeclareMathOperator*{\argmin}{arg\,min}

\newtheorem{problem}{Problem}
\newtheorem{assumption}{Assumption}
\newtheorem{proposition}{Proposition}
\newtheorem{definition}{Definition}
\newtheorem{theorem}{Theorem}
\newtheorem{lemma}{Lemma}
\newtheorem{corollary}{Corollary}
\newtheorem{remark}{Remark}

\begin{document}

\maketitle
\thispagestyle{empty}
\pagestyle{empty}

\begin{abstract}

In this paper, we investigate safety-critical control problem of discrete-time stochastic systems with incomplete information, where safety constraints must be enforced using state estimates obtained from noisy measurements. We develop an output-feedback control barrier function (CBF) framework based on an expectation-based discrete-time barrier condition that explicitly incorporates estimation uncertainty through the evolving belief over the state. To enable real-time implementation, we derive deterministic sufficient conditions that conservatively enforce the expectation-based CBF  by bounding the expectation with computable functions of the belief statistics  using Jensen inequalities. The resulting safety filter is formulated as a tractable optimization problem compatible with standard online controllers. Numerical simulations demonstrate that the proposed output-feedback approach achieves fast online computation while providing reliable safety performance in the presence of process noise and measurement uncertainty.

\end{abstract}

\section{Introduction}
Ensuring safety with formal guarantees is a central challenge in autonomous systems. In recent years, control barrier functions (CBFs), originally introduced in \cite{ames2016control}, have attracted considerable attention as a systematic framework for safe controller synthesis. By introducing a barrier function whose superlevel set characterizes the safe region, CBFs enable the design of control policies that render this set forward invariant, thereby ensuring that the system remains safe over the operation time.

Building on this framework, various extensions have been developed for systems subject to disturbances and uncertainty. For example, input-to-state safe CBFs were proposed in \cite{kolathaya2018input,alan2021safe} to explicitly characterize the effect of perturbations on safety. Beyond the state-feedback setting, observer-based CBF frameworks have also been developed for safe control under partial state information; see, e.g., \cite{agrawal2022safe,wang2024immersion,quan2025observer}. However, these approaches are mainly designed for deterministic systems or systems with bounded non-stochastic disturbances.
 
To address stochastic disturbances with unbounded support, stochastic extensions of CBFs have been developed for continuous-time systems. For example, \cite{clark2021control} proposed a stochastic CBF framework for systems with process and measurement noise, showing that safety can still be characterized under stochastic uncertainty. Along this line, \cite{nishimura2024control} derived sufficient conditions for almost-sure safety and introduced a stochastic zeroing CBF to characterize the probability that a sample path remains within the safe set. Moreover, \cite{wang2021safety} developed stochastic CBFs for nonlinear stochastic systems and used them to synthesize safety-critical controllers under diffusion-driven uncertainty.

However, the aforementioned works are mainly developed for continuous-time systems. In practice, many real-world systems are more naturally modeled and controlled in discrete time due to sampling, scheduling, digital implementation, online optimization, model predictive control, and discrete learning updates. This has motivated the development of discrete-time CBFs for safety-critical control \cite{agrawal2017discrete,zeng2021safety}. For stochastic systems, related discrete-time barrier-function-based results have also been reported, including finite-time probabilistic safety verification via barrier functions \cite{santoyo2019verification}, expectation-based stochastic discrete-time CBFs with Jensen-based reformulations and finite-horizon probability bounds \cite{cosner2023robust}. Nevertheless, these works are still largely limited to state-feedback settings or specific classes of safety constraints, which motivates the study of output-feedback safe control for discrete-time stochastic systems.

In this paper, we study safety-critical control for discrete-time stochastic systems with incomplete state information. To handle the output-feedback setting, we incorporate Kalman filtering into the CBF framework and construct an estimate-based safety function that links safety of the estimated state to that of the true state. On this basis, we develop a tractable deterministic reformulation of the stochastic CBF condition by combining an expectation-based barrier condition with a Jensen-type bound, and further derive explicit probabilistic safety guarantees over a finite horizon. 

It is worth noting that \cite{kishida2024risk} studies a related setting by integrating Kalman filtering, worst-case CVaR, and CBFs. In contrast, our approach is expectation-based and yields a Jensen-based deterministic reformulation. Since the worst-case CVaR method involves a min-max structure, it becomes much more expensive for complex safety constraints, while our approach directly gives a deterministic barrier constraint.

The main contributions are summarized as follows:
\begin{itemize}
    \item First, by formulating an expectation-based CBF condition on the Kalman filter state estimates and relating estimated-state safety to true-state safety, we establish an explicit probabilistic safety bound for output-feedback control. This is in contrast to \cite{kishida2024risk}, where no theoretical safety probability guarantee is provided.
    \item Second, by exploiting the sufficient-statistic property of the Kalman state estimate, we derive conditional moment properties of the estimation error and convert the expectation-based stochastic CBF condition into a tractable Jensen-based deterministic CBF condition.
    \item Finally, compared with \cite{kishida2024risk}, which only considers half-space and ellipsoidal safety constraints, our approach applies to more general barrier functions, provided that the barrier function is twice continuously differentiable, has a uniformly bounded Hessian, and admits a global upper bound. Simulations further show that our approach can attain the same empirical safety probability as \cite{kishida2024risk} with improved computational efficiency.
\end{itemize}

The remaining part is organized as follows. In Section~\ref{sec:problem}, we formulate the problem under investigation and review the expectation-based CBF framework. Section~\ref{sec:main} presents our solution to the output-feedback safe control problem. Section~\ref{sec:simulation} provides simulation studies to validate the proposed approach and compare it with existing related works. Finally, Section~\ref{sec:conclusion} concludes the paper.

\section{Preliminaries}\label{sec:problem}

\subsection{System Model}

Consider a discrete-time stochastic system\vspace{-3pt}
\begin{subequations}\label{eq-sys}
    \begin{align}
        x_{k+1}&=A_kx_k+B_ku_k+w_k\label{eq-state}\\
        y_k&=C_kx_k+v_k\vspace{-3pt}
    \end{align}
\end{subequations}
where $x_k\!\in\!\RR^n$ is the state, $u_k\!\in\!\RR^m$ is the control input, $w_k\!\in\!\RR^{n_w}$ is the disturbance, $y_k\!\in\!\RR^{n_y}$ is the measurement, $v_k\!\in\!\RR^{n_v}$ is the measurement noise at time instant $k\!\in\!\NN$. 
Given a control sequence $\u_{0:K-1}:=u_0u_1\cdots u_{K-1}$, we denote the corresponding system trajectory by\vspace{-3pt}
\begin{equation}
    \xi(x_0,\u_{0:K-1}):=x_0x_1\cdots x_K\notag\vspace{-3pt}
\end{equation}
where each state $x_k$ evolves according to \eqref{eq-state}. Owing to the stochastic disturbances, $\xi(x_0,\u_{0:K-1})$ is a random trajectory. Similarly, we denote by $\y_{0:K}:=y_0y_1\cdots y_K$ the measurement output sequence generated by \eqref{eq-sys}. 

Throughout this paper, we impose the following assumption on the stochastic processes.

\begin{assumption}\label{ass-gaussian}
The disturbances $w_k$ and the measurement noises $v_k$ are mutually independent zero-mean Gaussian white-noise processes with finite second moments.
\end{assumption}

In this paper, we characterize \emph{safety} by a function $h\!:\!\RR^n\!\to\!\RR$. In particular, $h(x_k)\!\geq\! 0$ indicates that the system is safe at time $k$, whereas $h(x_k)\!<\!0$ indicates that the system is unsafe. Since $x_k$ is a random variable due to the presence of stochastic disturbances, we adopt the following notion of \emph{$K$-step Exit Probability} \cite{cosner2023robust} to characterize the probability that the system leaves the safe set over a finite horizon $K$.
\begin{definition}[\textbf{$K$-step Exit Probability}]  
    Given a safety function $h\!:\!\RR^n\!\to\!\RR$ and the random trajectory $\xi(x_0,\u_{0:K-1})$ generated by  system \eqref{eq-sys} under a control input sequence $\u_{0:K-1}$ over horizon $K\!\in\!\NN$, the $K$-step exit probability is defined as\vspace{-3pt}
    \begin{equation}
        P_e^x:=P_e(\xi(x_0,\u_{0:K-1}))=\PP\left[\min_{k\in\{0,\ldots,K\}} h(x_k)< 0 \right].\vspace{-3pt}\notag
    \end{equation}
\end{definition}
Accordingly, the safety probability of the random trajectory $\xi(x_0,\u_{0:K-1})$ is defined as\vspace{-3pt}
\begin{equation}\label{eq-Ps}
    P_s^x:=P_s(\xi(x_0,\u_{0:K-1}))=1-P_e(\xi(x_0,\u_{0:K-1})).\vspace{-3pt}
\end{equation}

With the above preliminaries, now we formulate the main problem under investigation as follows.

\begin{problem}[\textbf{Output-Feedback Safe Control}]
Given system~\eqref{eq-sys} satisfying Assumption~\ref{ass-gaussian} and a nominal controller
$\hat{\mathbf u}_{0:T-1}:=\hat u_0\hat u_1\cdots\hat u_{T-1}$ for a finite horizon \(T\in\mathbb N\), synthesize an output-feedback controller of the form\vspace{-3pt}
\[
u_k=\pi_k(\mathbf y_{0:k},\hat{\mathbf u}_{0:k}), \qquad k=0,\dots,T-1,\vspace{-3pt}
\]
such that $\mathbf u_{0:T-1}:=u_0u_1\cdots u_{T-1}$
stays as close as possible to the nominal controller $\hat{\mathbf u}_{0:T-1}$ while guaranteeing safety over the entire horizon with probability at least \(1-\epsilon\), i.e.,
\begin{subequations}
\begin{align}
\min_{\mathbf u_{0:T-1}} \quad & \sum_{k=0}^{T-1}\|u_k-\hat u_k\|_2^2\\
\text{s.t.}\quad &
P_s(\xi(x_0,\u_{0:T-1}))\geq 1-\epsilon.
\end{align}
\end{subequations}
\end{problem}

\subsection{Expectation-based Control Barrier Function}

To enforce safety for the stochastic systems, we adopt the expectation-based discrete-time stochastic control barrier function framework \cite{cosner2023robust} with slight modifications. We first introduce several preliminary results as follows.

\begin{lemma}[\textbf{Discrete-Time Stochastic CBF} \cite{cosner2023robust}]\label{lem-expectationcbf}
Consider system~\eqref{eq-state} and a continuous function $h\!:\!\RR^n\!\to\!\RR$ satisfying\vspace{-3pt}
\begin{equation}
    \exists M>0,~\forall x\in\RR^n:h(x)\leq M.\vspace{-3pt}
\end{equation}
The function $h$ is called a discrete-time stochastic control barrier function if there exist $\alpha\in(0,1)$, $\delta\leq M(1-\alpha)$ and a control input sequence $\u_{0:K-1}$ such that\vspace{-3pt}
    \begin{equation}
        \EE[h(A_kx_k+B_ku_k)\mid x_k]\geq \alpha h(x_k)+\delta,~\forall x_k\in\RR^n.\vspace{-3pt}
    \end{equation}
Then, for any $K\!\in\!\NN$, the $K$-step exit probability of the closed-loop trajectory $P_e^x$ satisfies
\begin{itemize}
    \item if $\delta<0$, we have\vspace{-3pt}
    \begin{equation}
        P_e^x\leq \left(\frac{M-h(x_0)}{M}\right)\alpha^K+\frac{M(1-\alpha)-\delta}{M}\sum_{i=1}^K\alpha^{i-1}\notag\vspace{-3pt}
    \end{equation}
    \item if $\delta\geq 0$, we have\vspace{-3pt}
    \begin{equation}
         P_e^x \leq 1- \frac{h(x_0)}{M}\left(\frac{M\alpha+\delta}{M}\right)^K\notag\vspace{-3pt}
    \end{equation}
\end{itemize}
\end{lemma}

However, the above expectation-based CBF is generally intractable. To address this issue, we seek a deterministic condition that upper bounds the expectation condition in a tractable manner. For this purpose, we recall a Jensen-type inequality \cite{liao2019sharpening}, which characterizes the relation between the expectation of a function and the function evaluated at the expectation through a covariance-dependent correction term.

\begin{lemma}[\textbf{Jensen's Inequality}]\label{lem-jensen}
    Consider a function $h:\RR^n\to\RR$ and a random variable $x$ that takes values in $\RR^n$ satisfying
    \begin{itemize}
        \item [i)] $h$ is twice-continuously differentiable; 
        \item [ii)] there is $\lambda_{\max}\!>\!0$ such that $\sup_{x\in\RR^n}\!\|\nabla^2 h(x)\|_2\!\leq\!\lambda_{\max}$;
        \item [iii)] $\EE[\|x\|]<\infty$ and $\text{cov}(x)<\infty$.
    \end{itemize}
    Then we have\vspace{-3pt}
    \begin{equation}
        \EE[h(x)]\geq h(\EE[x])-\frac{1}{2}\lambda_{\max}\text{tr}(\text{cov}(x))\vspace{-3pt}\notag
    \end{equation}
    
\end{lemma}

\begin{proof}
Since $h$ is twice-continuously differentiable, based on the second-order Taylor theorem, we have\vspace{-3pt}
\begin{flalign}\label{eq-taylor}
    h(x)&=h(\EE[x])+\nabla h(\EE[x])^T (x-\EE[x])\notag\\
    &~~~+\frac{1}{2}(x-\EE[x])^T \nabla^2 h(\xi_x)(x-\EE[x])\vspace{-3pt}
\end{flalign}
where $\xi_x:=\EE[x] + \theta_x (x - \EE[x] )$ for some $\theta_x\in(0,1)$.\vspace{2pt}

By $\sup_{x\in\RR^n}\!\|\nabla^2 h(x)\|_2\!\leq\!\lambda_{\max}$, we have\vspace{-2pt}
    \begin{flalign}
        (x-\EE[x])^T \nabla^2 h(\xi_x)(x-\EE[x]) &\geq -\|\nabla^2 h(\xi_x)\left\|\|x\!-\!\EE[x]\right\|^2\notag\\
        &\geq -\lambda_{\max}\left\|x-\EE[x]\right\|^2\notag\vspace{-3pt}
    \end{flalign}
Taking expectation of both sides, we have\vspace{-3pt}
\begin{flalign}\label{eq-expectationtaylor}
    \EE[(x\!-\!\EE[x])^T \nabla^2 h(\xi_x)(x\!-\!\EE[x])]&\geq -\lambda_{\max}\EE\left[\|x\!-\!\EE[x]\|^2\right]\notag\\
    &=-\lambda_{\max}\text{tr}(\text{cov}(x))\vspace{-3pt}
\end{flalign}
Substituting \eqref{eq-expectationtaylor} into \eqref{eq-taylor}, we have\vspace{-3pt}
\begin{flalign}
    \EE[h(x)]&=h(\EE[x])+\frac{1}{2}\EE[(x-\EE[x])^T \nabla^2 h(\xi)(x-\EE[x])]\notag\\
    &\geq h(\EE[x])-\frac{1}{2}\lambda_{\max}\text{tr}(\text{cov}(x)) \notag \vspace{-3pt}
\end{flalign}
The proof is thus completed.
\end{proof}

\begin{remark}
Here we would like to remark that we mainly follow Lemma 1 in \cite{cosner2023robust}, but remove the concavity assumption on $h$. In fact, as shown in the proof above, the concavity of $h$ is not necessary for the stated bound to hold. 
\end{remark}

Now we present a tractable deterministic condition, derived from the Jensen's inequality, to enforce the expectation-based CBF condition. The following result adapts \cite[Theorem 6]{cosner2023robust} from general nonlinear systems to the time-varying linear system considered in this paper.

\begin{lemma}
    Consider system~\eqref{eq-state}. Let $h:\RR^n\to\RR$ be
    \begin{itemize}
        \item [i)] twice-continuously differentiable; and
        \item [ii)] there is $M\!>\!0$ such that $\sup_{x\in\RR^n}h(x)\leq M$.
    \end{itemize}
    Suppose that there exists $\alpha\in(0,1)$ such that\vspace{-3pt}
    \begin{equation}\label{eq-cbffully}
        h(A_kx_k+B_ku_k)-c_J\geq \alpha h(x_k),~\forall x_k\in\RR^n, \vspace{-3pt}
    \end{equation}
    where \vspace{-3pt}
    \begin{equation}
        0\leq c_J\leq \frac{\lambda_{\max}}{2}\text{tr}(\text{cov}[w_k])+M(1-\alpha) \vspace{-3pt}\notag
    \end{equation}
    Then we have $\delta=c_J-\frac{\lambda_{\max}}{2}\text{tr}(\text{cov}[w_k])$ such that\vspace{-3pt}
    \begin{equation}
        \EE\left[ h(A_kx_k+B_ku_k)\mid x_k \right]\geq \alpha h(x_k) + \delta,~\forall x_k\in \RR^n. \vspace{-3pt}\notag
    \end{equation}
\end{lemma}

The above results are established under the assumption of perfect state information. In contrast, when only partial and noisy output measurements are available, the exact state cannot be used to update the CBF condition \eqref{eq-cbffully}. To this end, in what follows, we introduce an output-feedback safe control framework that combines state estimation and stochastic barrier-based safety guarantees.

\section{Main Results}\label{sec:main}

In this section, we first review the Kalman filtering technique for state estimation and introduce the corresponding safety condition based on the estimated state. We then present our main result, which provides a tractable output-feedback safe controller based on the Jensen's inequality together with an explicit probabilistic safety guarantee.

\subsection{Kalman Filtering and Safety under State Estimation}

To handle the output-feedback control setting, we employ a Kalman filter \cite{anderson2005optimal} to estimate the state of system~\eqref{eq-sys}. Let
\[
\hat{x}_k := \EE[x_k \mid y_0,\dots,y_{k-1}]
\]
denote the one-step state predictor. The filter is given by
\vspace{-3pt}
\begin{equation}\label{eq-kalman}
    \hat{x}_{k+1}=A_k\hat{x}_k+B_ku_k+K_k(y_k-C_k\hat{x}_k) \vspace{-3pt}
\end{equation}
where the Kalman gain is
\begin{equation}
    K_k=A_kP_kC_k^\top(C_kP_kC_k^\top+R_k)^{-1},\notag
\end{equation}
and the prediction error covariance matrix $P_k$
satisfies
\vspace{-3pt}
\begin{flalign}
    P_{k+1}
    &=A_kP_kA_k^\top+Q_k \notag\\
    &\quad -A_kP_kC_k^\top(C_kP_kC_k^\top+R_k)^{-1}C_kP_kA_k^\top,\notag\\
    P_0&=\EE\!\left[(x_0-\hat{x}_0)(x_0-\hat{x}_0)^\top\right],\notag\\
    Q_k&=\EE[w_kw_k^\top],\qquad
    R_k=\EE[v_kv_k^\top].\notag \vspace{-3pt}
\end{flalign}

We present two basic properties of the state estimator concerning the mean and covariance of the estimation error \cite{anderson2005optimal}, as well as its boundedness over a finite horizon \cite{clark2021control}.

\begin{proposition}\label{prop-expcov}
Consider system~\eqref{eq-sys} with Kalman filter~\eqref{eq-kalman} under Assumption~\ref{ass-gaussian}. For all $k\in\NN$, we have \vspace{-3pt}
    \begin{equation}
        \EE[x_k-\hat{x}_k]=0, ~\text{cov}[x_k-\hat{x}_k\mid \y_{0:k}]=P_k \notag
    \end{equation}
\end{proposition}

\begin{proposition}\label{prop-kalman}
Consider system~\eqref{eq-sys} with Kalman filter~\eqref{eq-kalman} under Assumption~\ref{ass-gaussian}. For any $\sigma\!>\!0$, there exists $\gamma\!>\!0$ such that \vspace{-3pt}
\begin{equation}
    \PP\left[ \sup_{0\leq k \leq K} \|x_k-\hat{x}_k\|_2\leq \gamma \right] \geq 1-\sigma\notag
\end{equation}
\end{proposition}

In output-feedback setting, the exact state is not directly available, and therefore the CBF condition must be updated using the state estimate rather than the true state. To this end, we construct a new safety function whose nonnegativity implies the safety of the true state. Inspired by \cite{clark2021control}, we define \vspace{-3pt}
\begin{equation}
    h_\gamma  =  \sup\{ h(x)\!:\!\|x-x^0\|\leq \gamma \text{ for some }x^0\in h^{-1}(\{0\}) \}\notag
\end{equation}

\begin{lemma}[Lemma 4 in \cite{clark2021control}]\label{lem-ineqestimate}
    If $\|x_k-\hat{x}_k\|_2\leq \gamma$ and $h(\hat{x}_k)>h_\gamma$ hold, then we have $h(x_k)\geq 0$.
\end{lemma}
By the above construction, we define the state estimate-based safety function as \vspace{-3pt}
\begin{equation}\label{eq-hhat}
    \hat{h}(x)=h(x)-h_\gamma \vspace{-3pt}
\end{equation}
Accordingly, the $K$-step exit probability associated with \(\hat h\) is defined as \vspace{-3pt}
 \begin{equation}
    P_e^{\hat{x}}:=\PP\left[\min_{k\in\{0,\ldots,K\}} \hat{h}(\hat{x}_k)< 0 \right]\notag
\end{equation}

Based on the above results, we can relate the exit probability of the true state to that of the estimated state as follows.

\begin{proposition}\label{coro-kalman}
    Consider system~\eqref{eq-sys} with Kalman filter~\eqref{eq-kalman} under Assumption~\ref{ass-gaussian}. Then, we have \vspace{-3pt}
    \begin{equation}
        P_e^x\leq P_e^{\hat{x}}+\sigma\notag \vspace{-3pt}
    \end{equation}
    Equivalently, we have \vspace{-3pt}
    \begin{equation}
    P_s^x \geq 1-P_e^{\hat{x}}-\sigma\notag  
\end{equation}
\end{proposition}

\begin{proof}
   By Lemma~\ref{lem-ineqestimate}, we have \vspace{-3pt}
    \begin{flalign}
        &\left\{ \min_{k\in\{0,1,\ldots,K\}} \hat{h}(\hat{x}_k)\geq 0 \right\}\!\cap\!\left\{\!\sup_{k\in\{0,1,\ldots,K\}}\!\left\|x_k-\hat{x}_k\right\|_2\leq\gamma\!\right\}\notag\\
        &\subseteq\left\{ \min_{k\in\{0,1,\ldots,K\}}h(x_k)\geq 0 \right\}\notag \vspace{-3pt}
    \end{flalign}
    Equivalently, we have \vspace{-3pt}
    \begin{flalign}
        &\left\{ \min_{k\in\{0,1,\ldots,K\}}h(x_k)< 0 \right\}\subseteq\notag\\
        &\left\{ \min_{k\in\{0,1,\ldots,K\}} \hat{h}(\hat{x}_k)< 0 \right\}\!\cup\!\left\{\!\sup_{k\in\{0,1,\ldots,K\}}\left\|x_k-\hat{x}_k\right\|_2\!>\!\gamma\!\right\}\notag \vspace{-3pt}
    \end{flalign}
    Taking probability of both sides, we have \vspace{-3pt}
    \begin{equation}
        P_e^x\leq P_e^{\hat{x}}+\PP\left[ \sup_{k\in\{0,1,\ldots,K\}}\left\|x_k-\hat{x}_k\right\|_2>\gamma \right]\notag 
    \end{equation}
    By Proposition~\ref{prop-kalman}, we have $P_e^x\leq P_e^{\hat{x}}+\sigma$.
    By the definition of \eqref{eq-Ps}, we have $P_s^x \geq 1-P_e^{\hat{x}}-\sigma$.
\end{proof}

\subsection{Output-Feedback Stochastic Control Barrier Function}

In this subsection, we derive a tractable deterministic control barrier condition for system~\eqref{eq-sys} by using a Jensen-type inequality to enforce the expectation-based CBF condition formulated in terms of the Kalman filter state estimate.

Before presenting the Jensen-based CBF condition, we first introduce a lemma showing that, for a linear time-varying system~\eqref{eq-sys}, conditioning on the current estimate $\hat{x}_k$ is sufficient to characterize the conditional covariance of the Kalman estimation error, without explicitly conditioning on the full measurement history $\y_{0:k}$.

\begin{lemma}\label{lem-kalman-cov}
Consider system~\eqref{eq-sys} with Kalman filter~\eqref{eq-kalman} under Assumption~\ref{ass-gaussian}. Then, for all $k\in\NN$, we have \vspace{-4pt}
\begin{flalign}
    \EE [x_k-\hat{x}_k\mid \hat{x}_k]&=0,\notag\\
    \cov[x_k-\hat{x}_k\mid \hat{x}_k]&=P_k. \vspace{-8pt}\notag
\end{flalign}
\end{lemma}

\begin{proof}
Under the linear-Gaussian setting, the Kalman filter yields the posterior distribution \vspace{-4pt}
\[
x_k\mid \y_{0:k}\sim\mathcal N(\hat x_k,P_k).  \vspace{-4pt}
\]
Since $P_k$ is known from the Riccati recursion, this Gaussian posterior is completely characterized by its mean $\hat x_k$. Hence, $\hat x_k$ is a sufficient statistic of $\y_{0:k}$ for the posterior distribution of $x_k$. Therefore, $x_k\mid \hat x_k\sim\mathcal N(\hat x_k,P_k)$, which implies \vspace{-3pt}
\[
x_k-\hat x_k\mid \hat x_k\sim\mathcal N(0,P_k).  \vspace{-3pt}
\]
The proof is thus completed.
\end{proof}

Based on the above result, now we present a Jensen-based tractable sufficient condition for the expectation-based CBF under Kalman-filtered state estimates.

\begin{proposition}\label{kalman-jensen}
    Consider system~\eqref{eq-sys} with Kalman filter based estimator \eqref{eq-kalman}. Consider a function $h:\RR^n\to\RR$ satisfying
    \begin{itemize}
        \item [i)] $h$ is twice-continuously differentiable;
        \item [ii)] there is $\lambda_{\max}\!>\!0$ such that $\sup_{x\in\RR^n}\!\|\nabla^2 h(x)\|_2\!\leq\!\lambda_{\max}$.
    \end{itemize}
    If there exists $c_J'\geq0$ such that \vspace{-3pt}
    \begin{equation}\label{eq-kalmancbf}
        h(A_k\hat{x}_k+B_ku_k)-c_J'\geq \alpha h(\hat{x}_k),~\forall x\in\RR^n \vspace{-3pt}
    \end{equation}
    then we have \vspace{-3pt}
    \begin{equation}
        \EE\left[ h(\hat{x}_{k+1})\mid \hat{x}_k \right]\geq \alpha h(\hat{x}_k)+\delta' \vspace{-3pt}
    \end{equation}
    where \vspace{-3pt}
    \begin{flalign}\label{eq-delta}
        \delta'\leq c_J'-\frac{\lambda_{\max}}{2}\text{tr}(K_k\text{cov}[v_k]K_k^\top) \!-\!\frac{\lambda_{\max}}{2}\text{tr}(K_kC_kP_kC_k^\top K_k^\top) \vspace{-3pt}
    \end{flalign}
\end{proposition}

\begin{proof}
By Lemma~\ref{lem-jensen}, we have
    \begin{flalign}
        &\EE[h(\hat{x}_{k+1})\mid\hat{x}_k]\notag\\
        &=\EE[h(A_k\hat{x}_k+B_ku_k+K_kv_k+K_k C_k(x_k-\hat{x}_k))\mid\hat{x}_k]\notag\\
        &\geq h(\EE[A_k\hat{x}_k+B_ku_k+K_kv_k+K_k C_k(x_k-\hat{x}_k)\mid \hat{x}_k])\notag\\
        &-\!\frac{\lambda_{\max}}{2}\text{tr}(\text{cov}(A_k\hat{x}_k\!+\!B_ku_k\!+\!K_kv_k\!+\!K_k C_k(x_k\!-\!\hat{x}_k)\mid \hat{x}_k))\notag
    \end{flalign}
Given $\hat{x}_k$, $A_k\hat{x}_k+B_ku_k$ is deterministic instead of a random variable. Therefore, we have \vspace{-3pt}
\begin{flalign}
    \EE[A_k\hat{x}_k+B_ku_k+K_kv_k\mid \hat{x}_k] &= A_k\hat{x}_k+B_ku_k\notag\\
    \cov[A_k\hat{x}_k+B_ku_k+K_kv_k\mid \hat{x}_k] &= A_k\hat{x}_k+B_ku_k\notag
\end{flalign}
By Proposition~\ref{prop-expcov} and Lemma~\ref{lem-kalman-cov}, we know that \vspace{-3pt}
\begin{equation}
     \EE[K_kC_k(x_k-\hat{x}_k)\mid \hat{x}_k] = 0\notag \vspace{-3pt}
\end{equation}
\begin{equation}
    \text{cov}[K_k C_k(x_k-\hat{x}_k)\mid \hat{x}_k]=K_kC_kP_kC_k^\top K_k^\top\notag \vspace{-3pt}
\end{equation}
Then we have \vspace{-3pt}
    \begin{flalign}
        &\EE[h(\hat{x}_{k+1})\mid\hat{x}_k]\geq h(A_k\hat{x}_k+B_ku_k)  \notag\\
        &-\frac{\lambda_{\max}}{2}\text{tr}(K_k\text{cov}[v_k]K_k^\top)-\frac{\lambda_{\max}}{2}\text{tr}(K_kC_kP_kC_k^\top K_k^\top )\notag\\
        &=h(A_k\hat{x}_k+B_ku_k)-c_J'+\delta'\notag\\
        &\geq \alpha h(\hat{x}_k)+\delta'\notag \vspace{-3pt}
    \end{flalign}
    The proof is thus completed.
\end{proof}

Finally, we present the main result, which combines the Kalman filter-based state estimation and the Jensen-based CBF reformulation to obtain an output-feedback controller with a probabilistic safety guarantee.

\begin{theorem}\label{thm}
    Consider system~\eqref{eq-sys} with Kalman filter~\eqref{eq-kalman}. Suppose $h:\RR^n\to\RR$ satisfies
    \begin{itemize}
        \item [i)] $h$ is twice-continuously differentiable;
        \item [ii)] there is $\lambda_{\max}\!>\!0$ such that $\sup_{x\in\RR^n}\!\|\nabla^2 h(x)\|_2\!\leq\!\lambda_{\max}$;
        \item [iii)] $h(x)\leq M,~\forall x\in\RR^n,~M>0$;
    \end{itemize}
    Given an initial state $x_0$, a nominal controller $\hat{\u}_{0:T-1}$ for a finite horizon $T\!\in\!\NN$, and $\hat{h}:\RR^n\to\RR$ defined in \eqref{eq-hhat}, the following controller \vspace{-5pt} 
    \begin{subequations}\label{eq-thm}
    \begin{align}
        &\u_{0:T-1}^*=\argmin_{u_t\in\RR^m}\sum_{k=0}^{T-1}\|u_k-\hat{u}_k\|_2^2\\
        \text{s.t. }& \hat{h}(A_k\hat{x}_k+B_ku_k)-c_J'\geq\alpha\hat{h}(\hat{x}_k)\label{eq-jensenhat}\\
        & 0\leq c_J'\leq (M-h_\gamma)(1-\alpha)+\frac{\lambda_{\max}}{2}\text{tr}(K_k\text{cov}[v_k]K_k^\top)\notag\\
        &~~~~~~~~~~~+\frac{\lambda_{\max}}{2}\text{tr}(K_kC_kP_kC_k^\top K_k^\top)\label{eq-cjprime}\\
        &0<\alpha<1\label{eq-alpha} \vspace{-3pt}
    \end{align}
    \end{subequations}
    makes the following hold: \vspace{-3pt}
    \begin{equation}\label{eq-Psgeq}
        P_s(\xi(x_0,\u_{0:T-1}^*))\geq 1- P_e^{\hat{x}}-\sigma \vspace{-3pt}
    \end{equation}
    where $P_e^{\hat{x}}$ satisfies
    \begin{itemize}
        \item if $\delta'<0$, we have \vspace{-3pt}
        \begin{equation}\label{eq-Pehat1}
        \!\!\!\!\!\!\!\!\!\!\!\!P_e^{\hat{x}}\!\leq\! \left(\!\frac{M\!-\!h_\gamma\!-\!\hat{h}(\hat{x}_0)}{M\!-\!h_\gamma}\!\right)\alpha^T\!+\!\frac{(M\!-\!h_\gamma)(1\!-\!\alpha)-\delta'}{M\!-\!h_\gamma}\!\sum_{i=1}^T\alpha^{i-1} \vspace{-3pt}
    \end{equation}
        \item if $\delta'\geq 0$, we have \vspace{-3pt}
        \begin{equation}\label{eq-Pehat2}
         P_e^{\hat{x}} \leq 1- \frac{\hat{h}(\hat{x}_0)}{M-h_\gamma}\left(\frac{(M-h_\gamma)\alpha+\delta'}{M-h_\gamma}\right)^T \vspace{-3pt}
    \end{equation}
    \end{itemize}
    with $\delta'$ defined in \eqref{eq-delta}.
\end{theorem}

\begin{proof}
    By Proposition~\ref{kalman-jensen}, we know that, under conditions i) and ii) and \eqref{eq-jensenhat}--\eqref{eq-alpha}, we have \vspace{-3pt}
    \begin{equation}
        \EE\left[ \hat{h}(\hat{x}_{k+1})\mid \hat{x}_k \right]\geq \alpha \hat{h}(\hat{x}_k)+\delta' \vspace{-3pt}
    \end{equation}
   Under condition iii),  by Lemma~\ref{lem-expectationcbf} in which $\hat{h}$ satisfies $\hat{h}(x)\leq M-h_\gamma,~\forall x\in\RR^n$ and  $\delta'$ in \eqref{eq-delta} with \eqref{eq-cjprime} satisfies \vspace{-3pt}
   \begin{equation}
       \delta'\leq (M-h_\gamma)(1-\alpha) \vspace{-3pt}
   \end{equation}
   we know that both \eqref{eq-Pehat1} and \eqref{eq-Pehat2} hold. 
   Then, by Proposition~\ref{coro-kalman}, we know that \eqref{eq-Psgeq} holds.
\end{proof}

\begin{remark}
Beyond the conditions given in Theorem~\ref{thm}, if we further assume that $h$ is concave, then the constraints induced by \eqref{eq-jensenhat} are convex with respect to the decision variables. Therefore, the optimization problem in \eqref{eq-thm} is convex. This ensures that any feasible local optimum is also a global optimum, and enables the use of efficient convex optimization algorithms with stronger guarantees on numerical reliability and computational tractability.
Nevertheless, if $h$ is not concave, the optimization problems may become nonconvex. In this case, Theorem~\ref{thm} still holds as long as the corresponding optimization problems are solved and their constraints are satisfied. However, one can no longer guarantee that the obtained solution is globally optimal, and the numerical solution may depend on specific solvers used.
\end{remark}

\section{Simulation}\label{sec:simulation}

In this section, we validate our approach through two sets of simulations. We first compare our method with the Kalman filter-based worst-case CVaR-CBF approach \cite{kishida2024risk}, with emphasis on computational time under matched empirical safety probabilities. We then adopt the inverted-pendulum example from \cite{cosner2023robust} to examine the proposed probabilistic safety bound, showing that the theoretical bound is providing a reliable characterization of the actual safety performance.

\subsection{Computational Comparison}

We consider the same linear system with \cite{kishida2024risk} \vspace{-3pt}
\begin{subequations}
    \begin{align}
        x_{t+1}&=\begin{bmatrix}
            1&0.05\\
            0&1
        \end{bmatrix}x_t+\begin{bmatrix}
            0.0125\\
            0.05
        \end{bmatrix}u_t+w_t\\
        y_t&=\begin{bmatrix}
        0&1
    \end{bmatrix}x_t+v_t \vspace{-5pt}
    \end{align} 
\end{subequations}
The disturbance and noise covariances are
\begin{flalign}
    Q=\begin{bmatrix}
        7.66\times 10^{-5}&3.06\times 10^{-3}\\
        3.06\times 10^{-3}&1.23\times 10^{-1}
    \end{bmatrix},\notag~R=0.09\notag \vspace{-5pt}
\end{flalign}
The nominal controller is given by \vspace{-3pt}
\begin{equation}
    u(x)=-15x_1-5x_2. \vspace{-3pt}\notag
\end{equation}

Since the Kalman filter-based CVaR-CBF approach proposed in \cite{kishida2024risk} is developed only for half-space and ellipsoidal safety sets, in what follows, we focus our comparison on these two types of safety constraints.

\textbf{Comparisons under the half-plane safety constraint.} We first consider the half-space safety set  \vspace{-3pt}
\begin{equation}
    \mathcal{C}_1=\{x\in\RR^2:h(x)=\begin{bmatrix}
        0.4&0.4
    \end{bmatrix}x+1\geq 0\}, \vspace{-3pt}\notag
\end{equation}
since the simulation example in \cite{kishida2024risk} is conducted only for a half-space safety constraint.
We adopt the initial conditions
$x_0=\begin{bmatrix}
    7&0
\end{bmatrix}^\top$, $P_0=Q$
and the parameters $\alpha=0.7$, $\epsilon=0.3$ from \cite{kishida2024risk} to reproduce its simulation setting. We then evaluate the two approaches over 100 Monte Carlo trials by generating 100 realizations of the stochastic disturbances and measurement noises, and record the proportion of trajectories that remain safe over the entire horizon as an estimate of $P_s$, which we denote by $\hat{P}_s$. For the method in \cite{kishida2024risk}, we obtain $\hat{P}_s=0.76$.

\begin{figure}
    \centering
    \includegraphics[width=\linewidth]{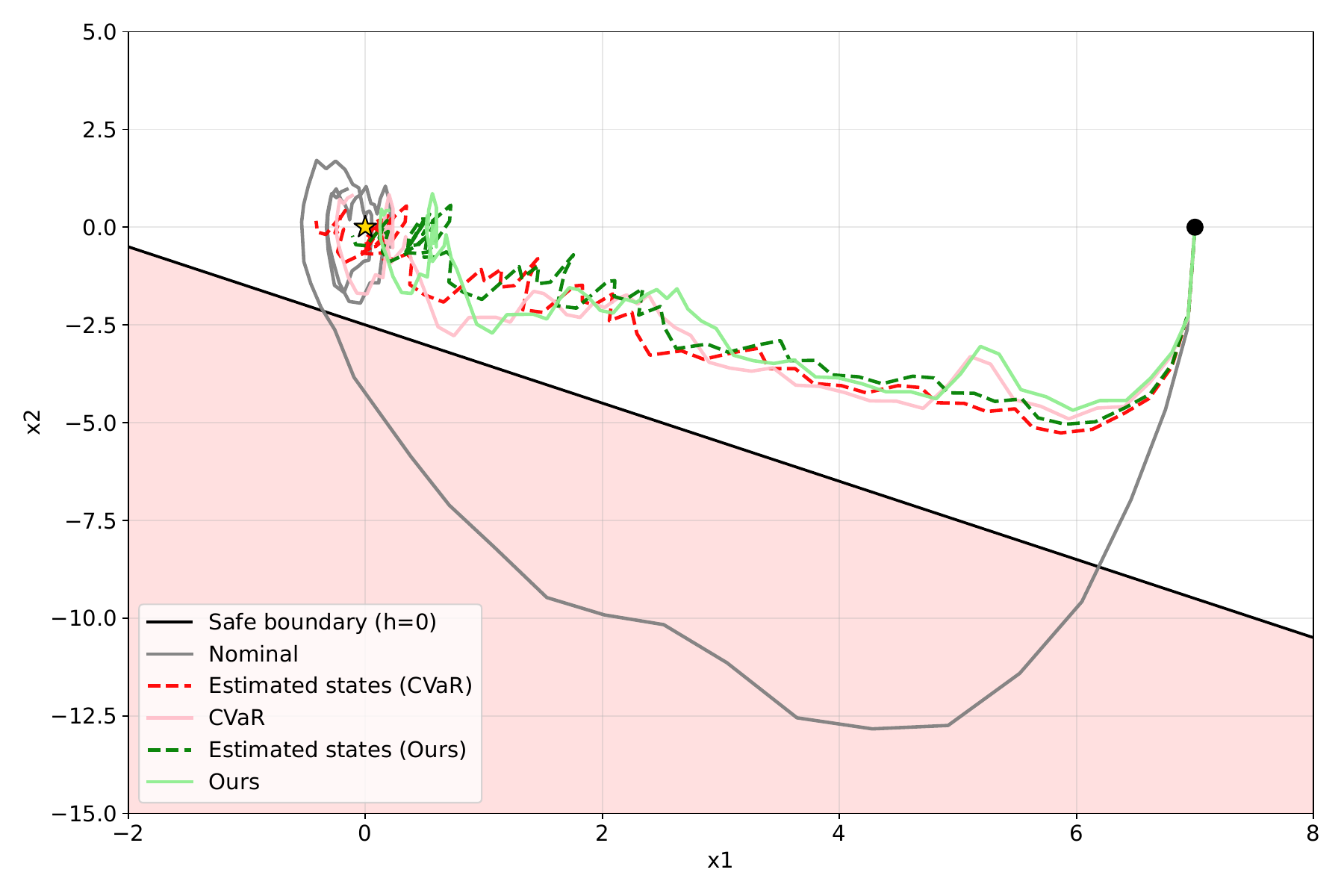}
    \caption{One realization of the closed-loop trajectories by the two methods under the half-space safety constraint. Both methods achieve the same empirical safety probability $\hat{P}_s=0.76$, with parameters $\alpha=0.7$ and $\epsilon=0.3$ for the method in \cite{kishida2024risk}, and $\alpha=0.7$ and $c_J'=0.115\times c_J^{\max}$ for our approach.} \vspace{-6pt}
    \label{fig:halfplane-comparison}
\end{figure}

\begin{figure*}[htbp]
    \centering
    \begin{subfigure}[b]{0.24\textwidth}
        \centering
        \includegraphics[width=\textwidth]{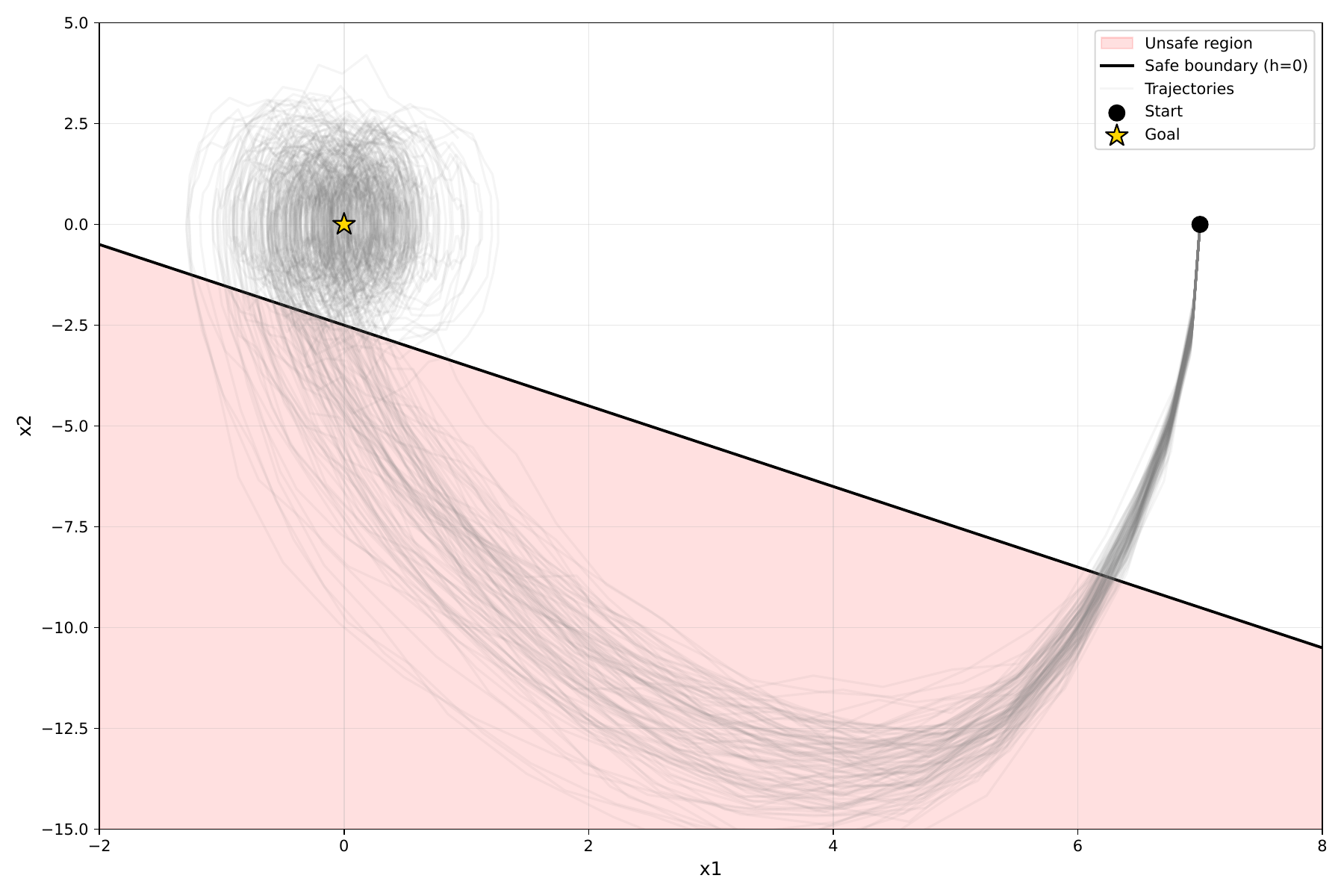}
        \caption{Nominal}
        \label{fig:nominal-halfplan}
    \end{subfigure}
    \hfill
    \begin{subfigure}[b]{0.24\textwidth}
        \centering
        \includegraphics[width=\textwidth]{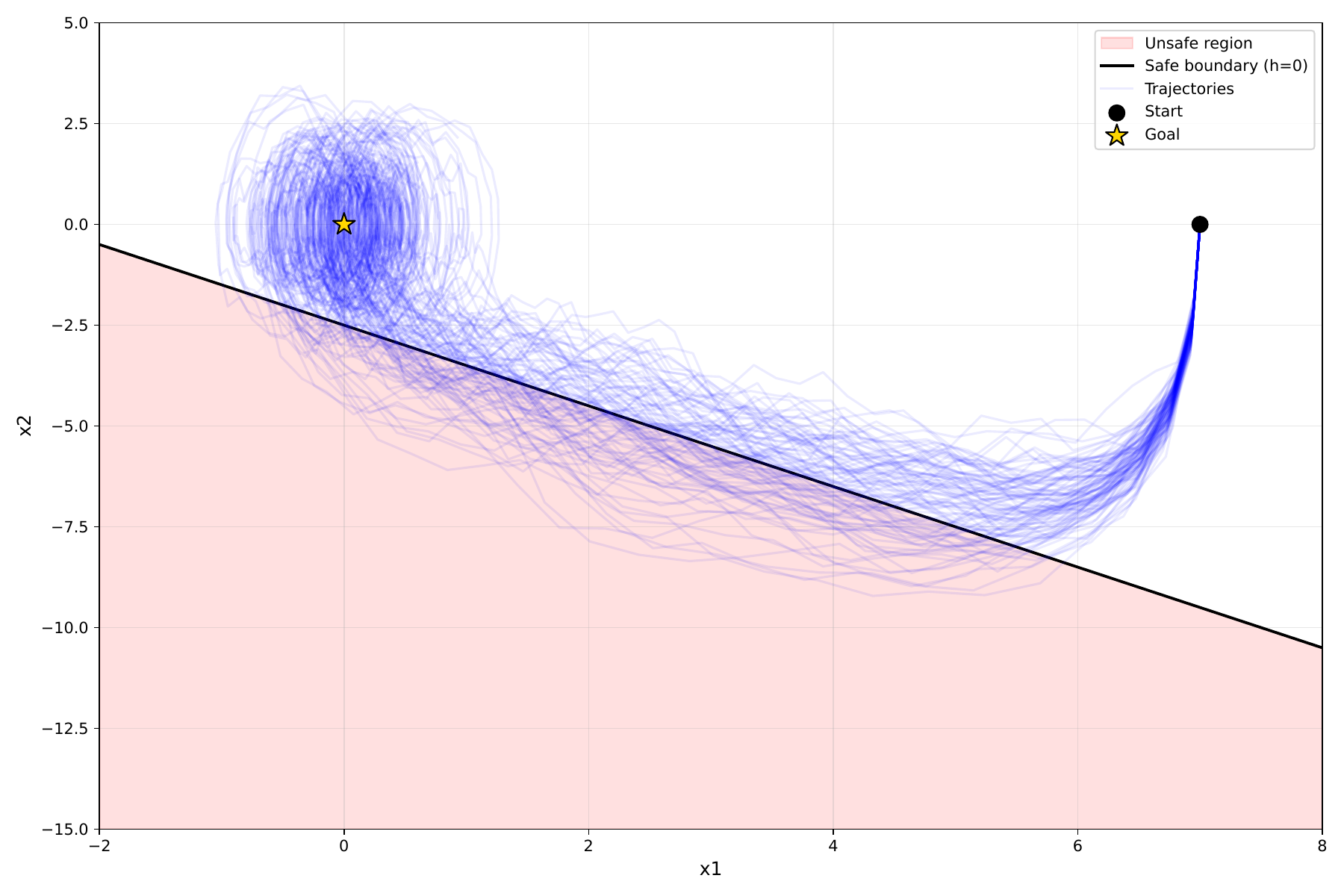}
        \caption{Nominal CBF}
        \label{fig:nominal_cbf-halfplan}
    \end{subfigure}
    \hfill
    \begin{subfigure}[b]{0.24\textwidth}
        \centering
        \includegraphics[width=\textwidth]{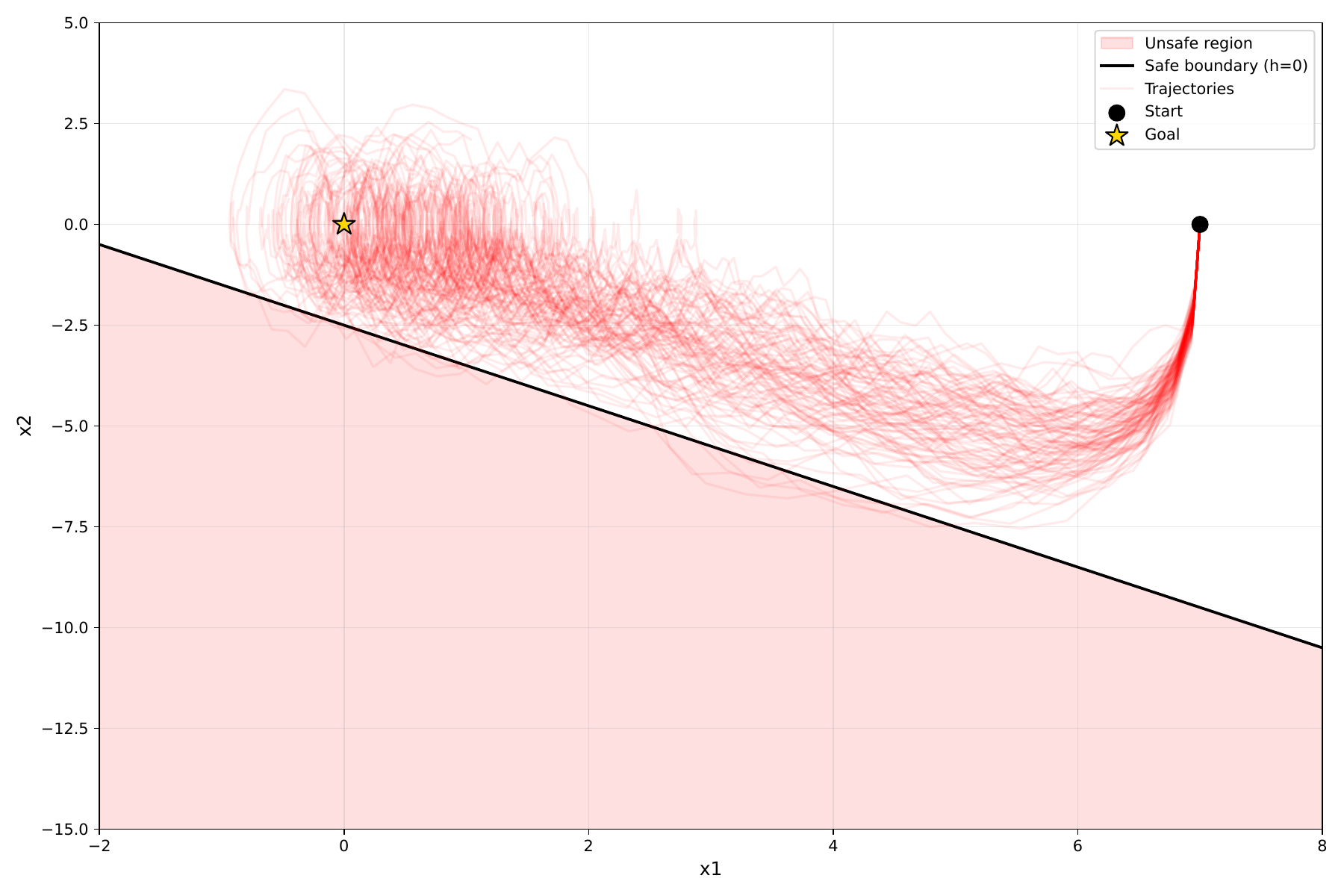}
        \caption{Kalman-CVaR-CBF}
        \label{fig:cvar-halfplan}
    \end{subfigure}
    \hfill
    \begin{subfigure}[b]{0.24\textwidth}
        \centering
        \includegraphics[width=\textwidth]{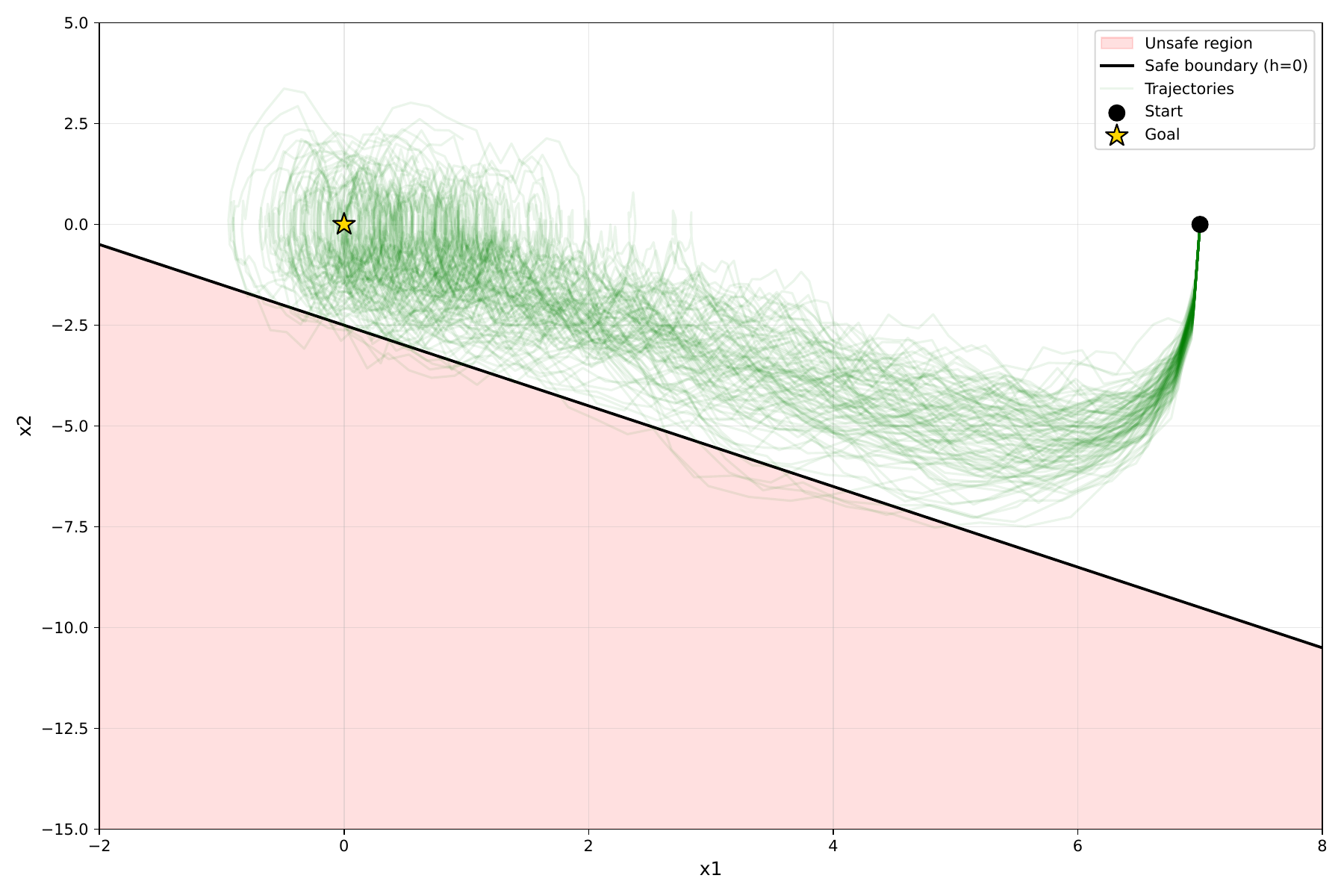}
        \caption{Ours}
        \label{fig:scbf-halfplan}
    \end{subfigure}

    \centering
    \begin{subfigure}[b]{0.24\textwidth}
        \centering
        \includegraphics[width=\textwidth]{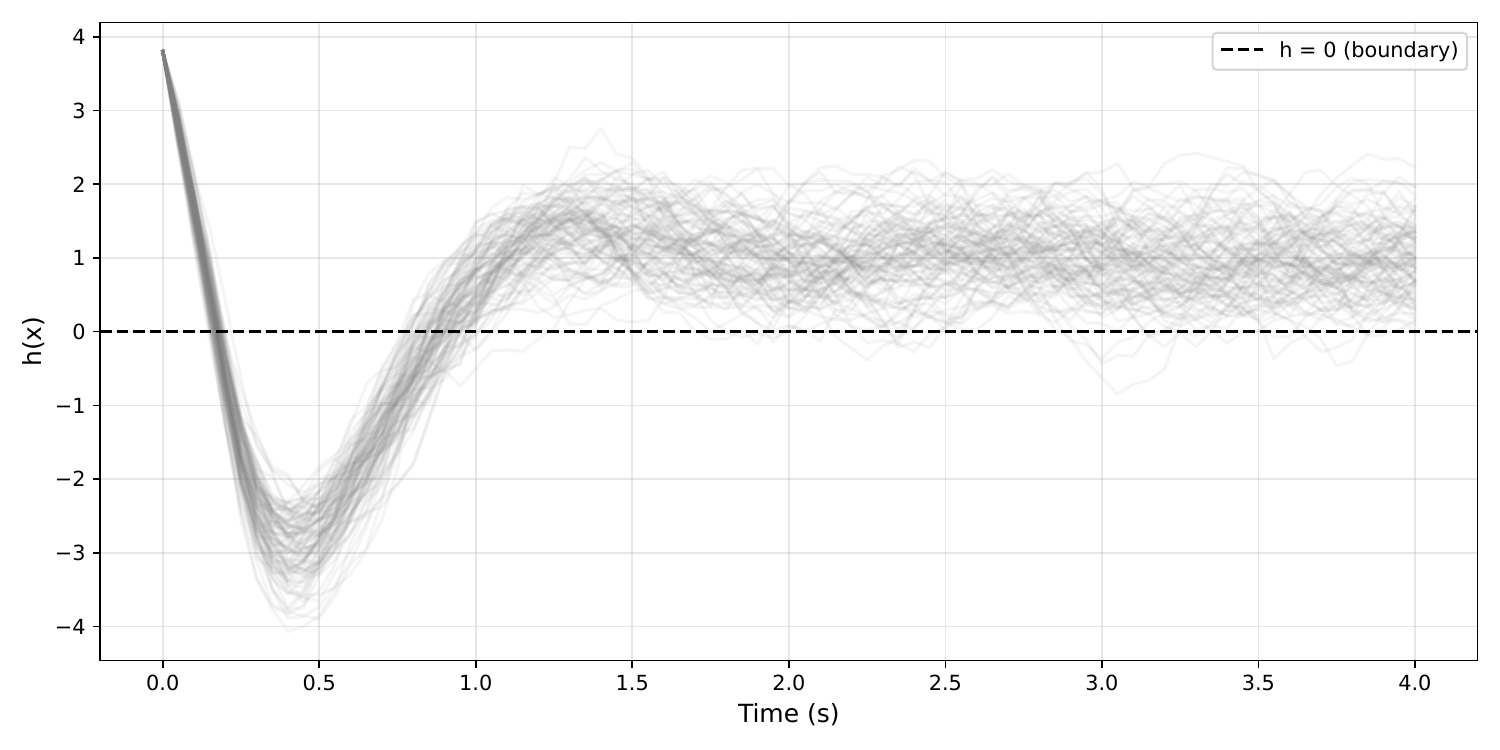}
        \caption{$h(x_k)$: Nominal}
        \label{fig:nominal-h}
    \end{subfigure}
    \hfill
    \begin{subfigure}[b]{0.24\textwidth}
        \centering
        \includegraphics[width=\textwidth]{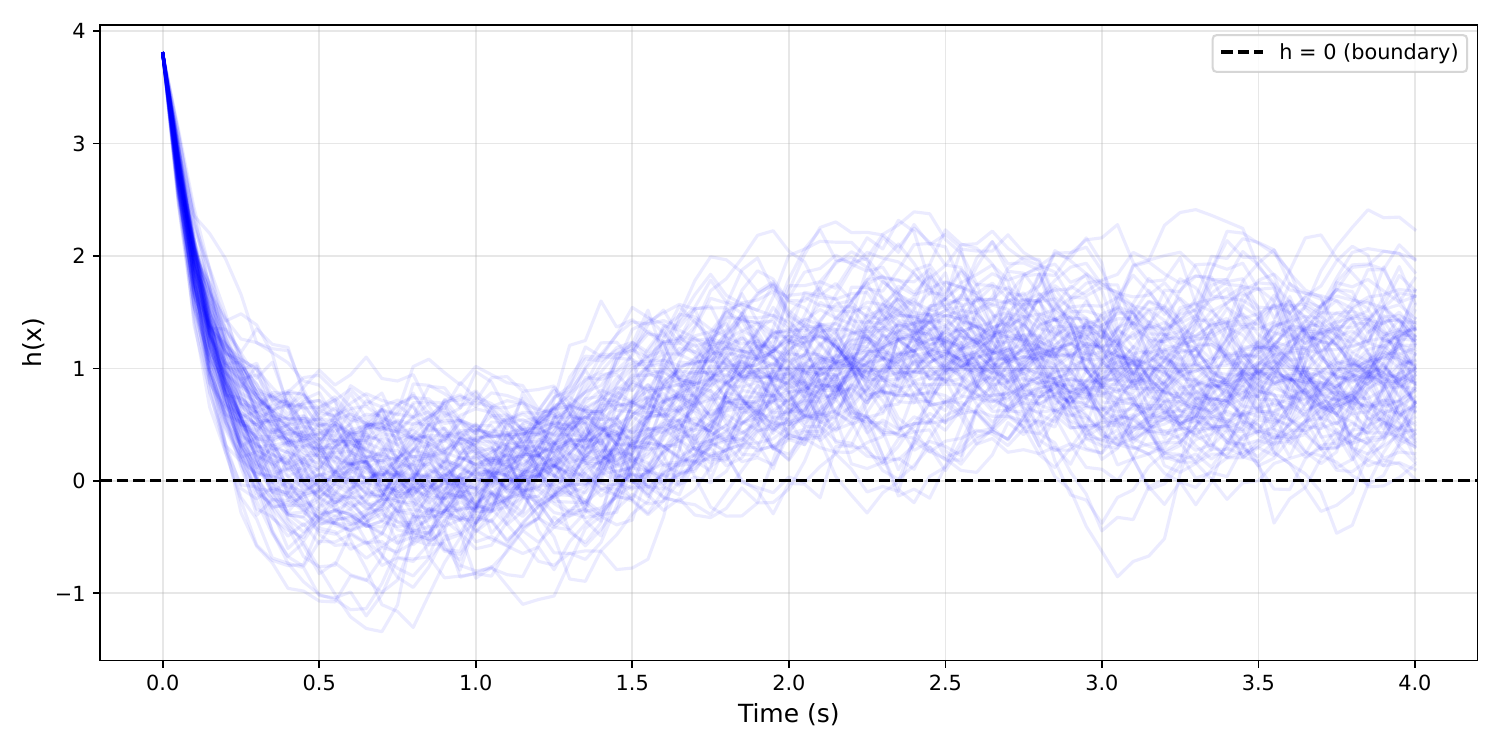}
        \caption{$h(x_k)$: Nominal CBF}
        \label{fig:nominal_cbf-h}
    \end{subfigure}
    \hfill
    \begin{subfigure}[b]{0.24\textwidth}
        \centering
        \includegraphics[width=\textwidth]{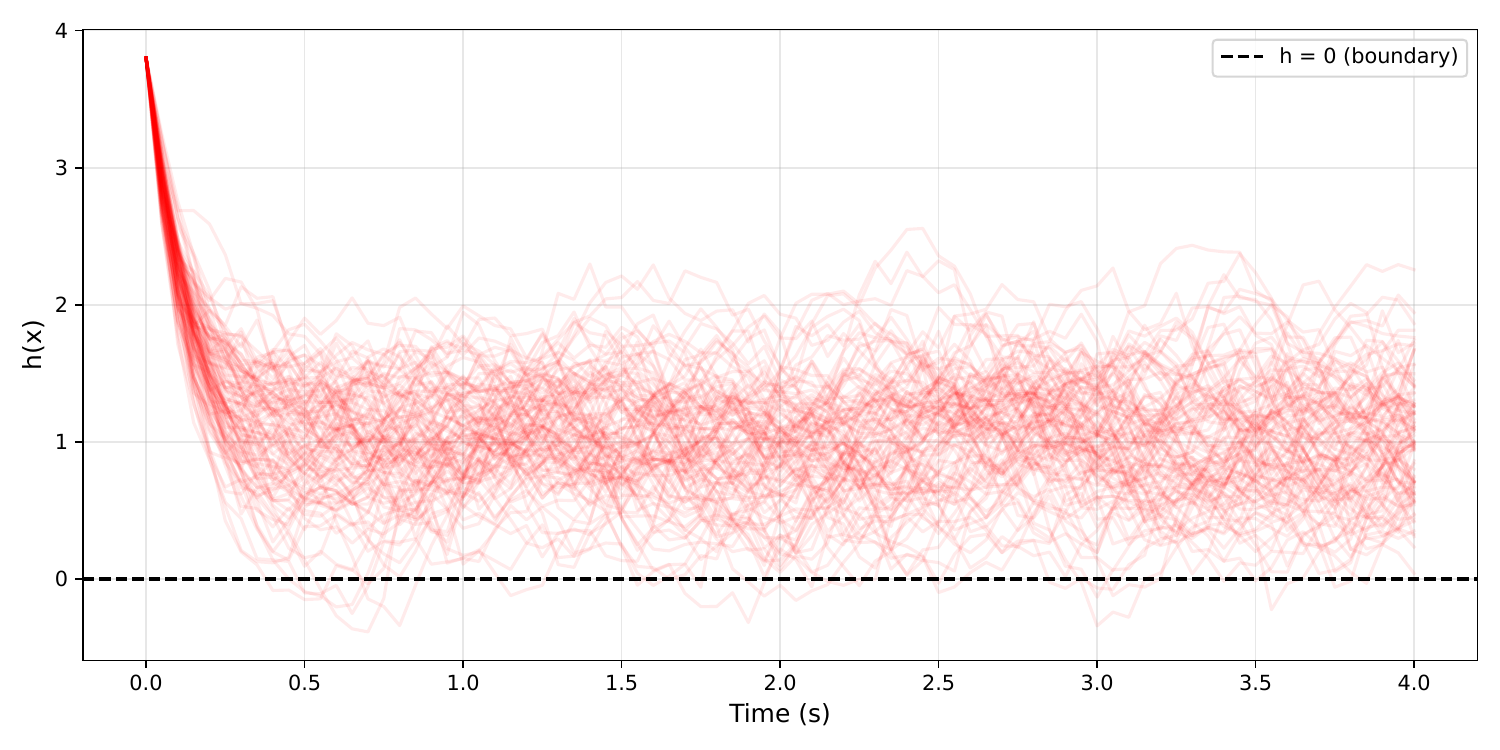}
        \caption{$h(x_k)$: Kalman-CVaR-CBF}
        \label{fig:cvar-h}
    \end{subfigure}
    \hfill
    \begin{subfigure}[b]{0.24\textwidth}
        \centering
        \includegraphics[width=\textwidth]{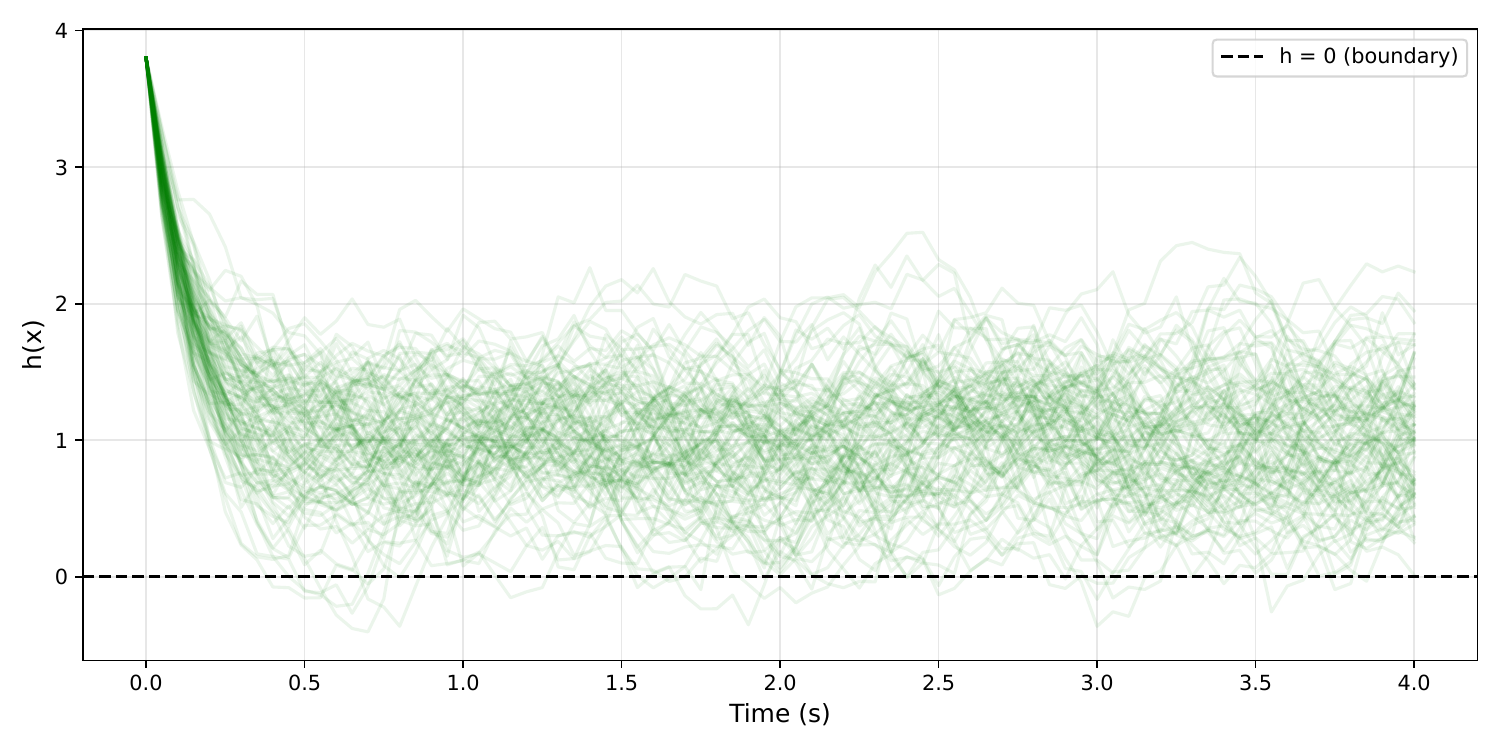}
        \caption{$h(x_k)$: Ours}
        \label{fig:scbf-h}
    \end{subfigure}
    
    \caption{Closed-loop state trajectories and  evolution of the safety function $h(x_k)$ over 100 Monte Carlo trials under the half-space safety constraint. The empirical safety probabilities of the four controllers are $0$, $0.05$, $0.76$, and $0.76$, respectively.}
    \label{fig:halfplan_all}
    
\end{figure*}

\begin{figure*}[htbp]
    \centering
    \begin{subfigure}[b]{0.24\textwidth}
        \centering
        \includegraphics[width=\textwidth]{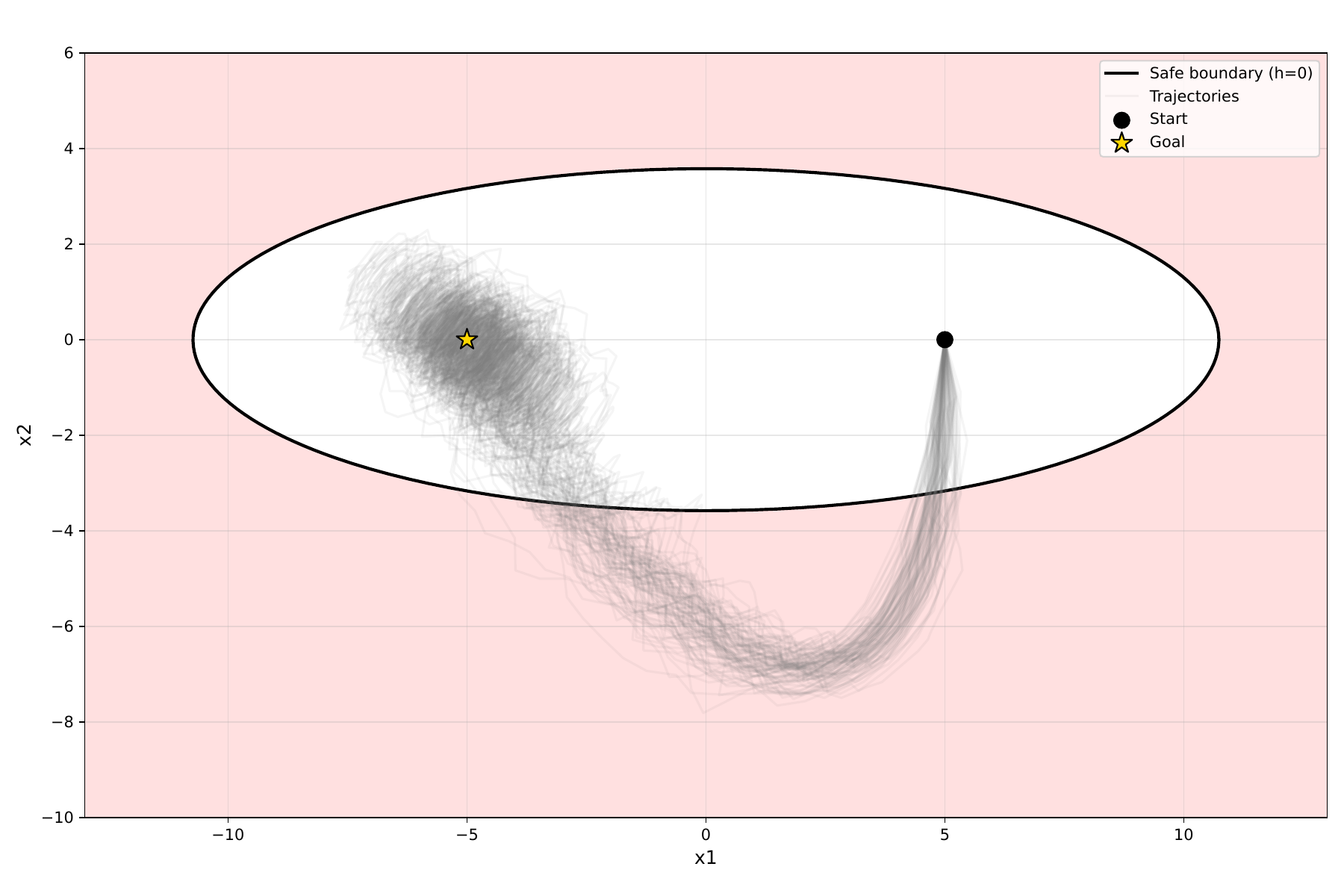}
        \caption{Nominal}
        \label{fig:nominal-oval}
    \end{subfigure}
    \hfill
    \begin{subfigure}[b]{0.24\textwidth}
        \centering
        \includegraphics[width=\textwidth]{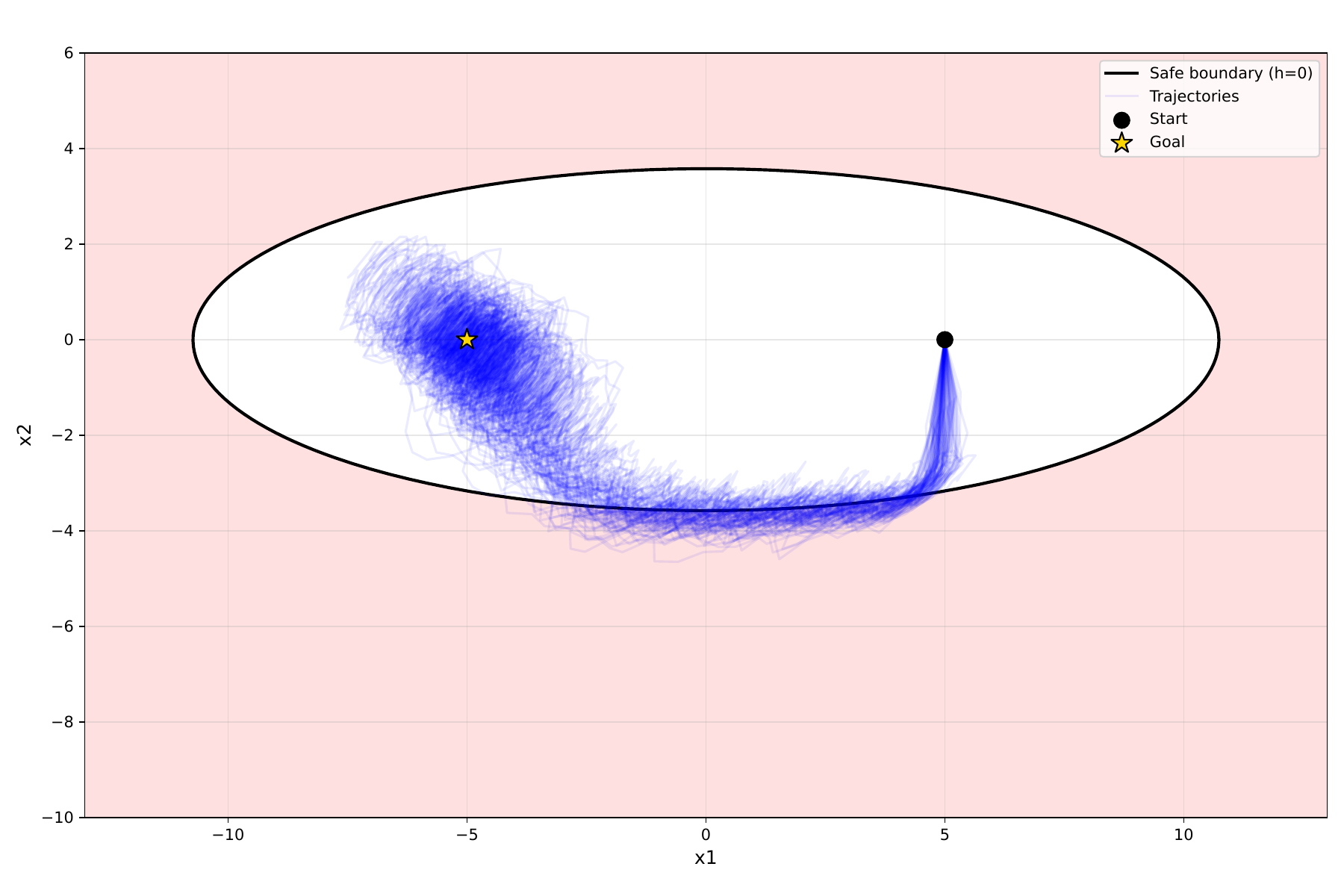}
        \caption{Nominal CBF}
        \label{fig:nominal_cbf-oval}
    \end{subfigure}
    \hfill
    \begin{subfigure}[b]{0.24\textwidth}
        \centering
        \includegraphics[width=\textwidth]{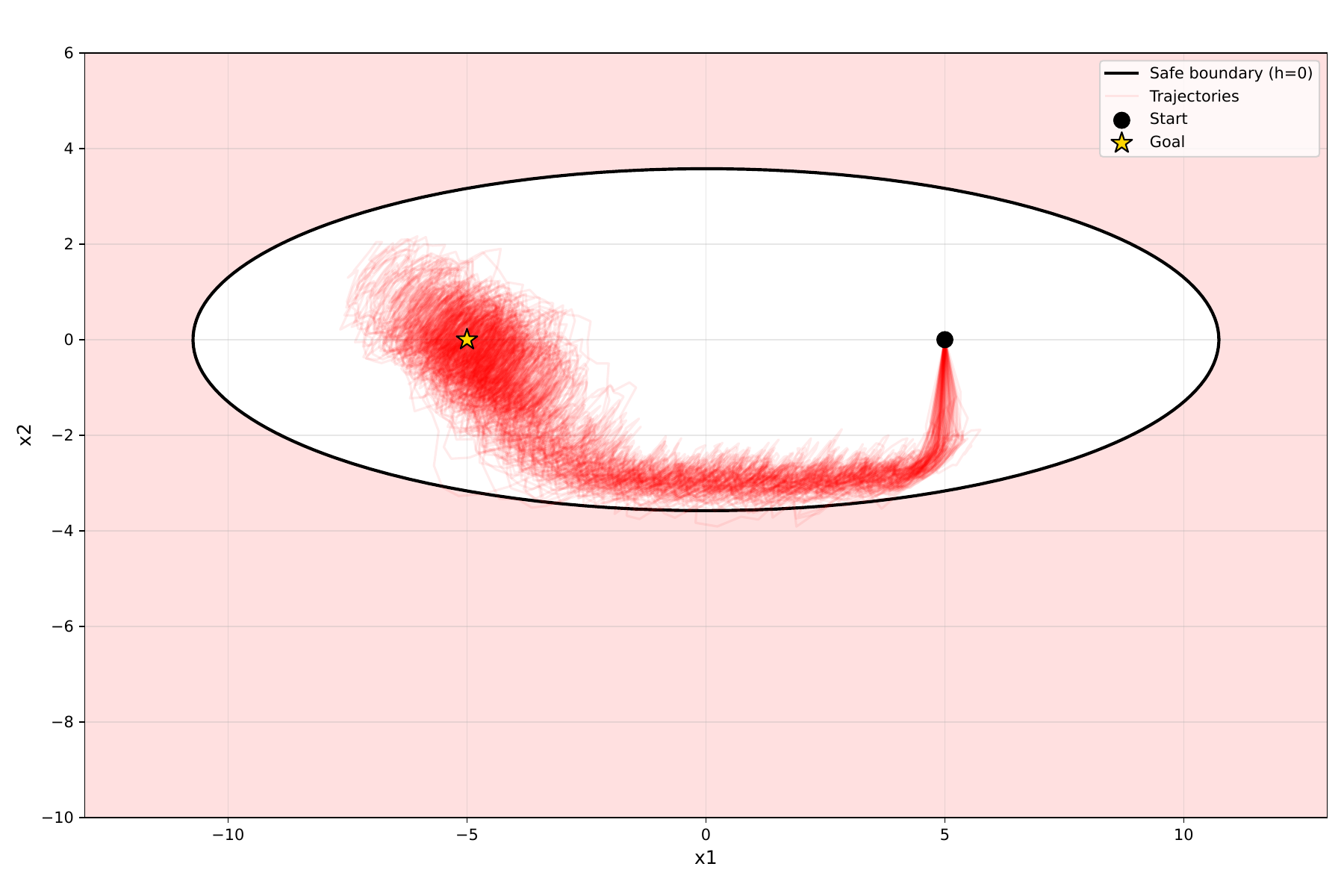}
        \caption{Kalman-CVaR-CBF}
        \label{fig:scbf-oval}
    \end{subfigure}
    \hfill
    \begin{subfigure}[b]{0.24\textwidth}
        \centering
        \includegraphics[width=\textwidth]{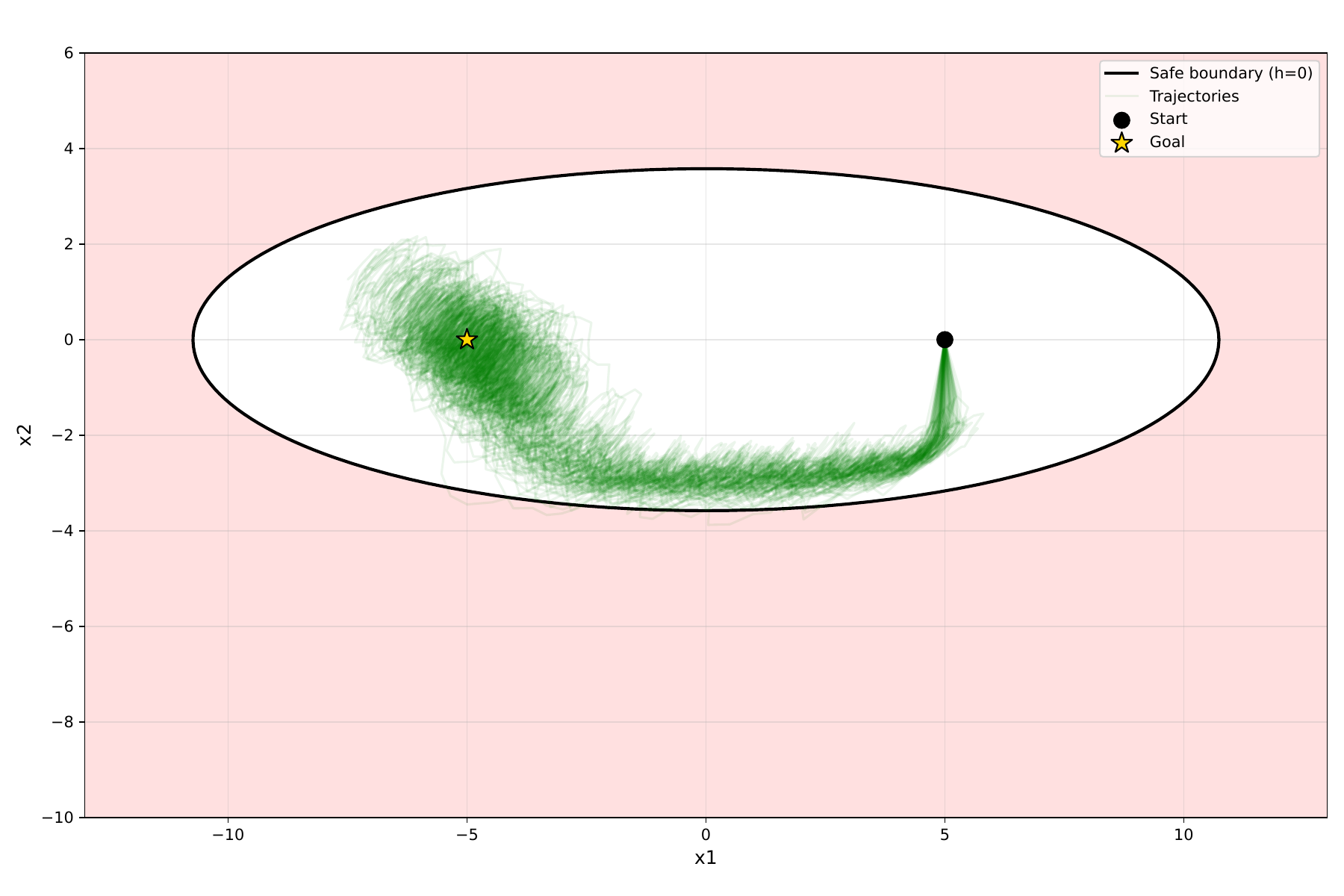}
        \caption{Ours}
        \label{fig:cvar-oval}
    \end{subfigure}

    \centering
    \begin{subfigure}[b]{0.24\textwidth}
        \centering
        \includegraphics[width=\textwidth]{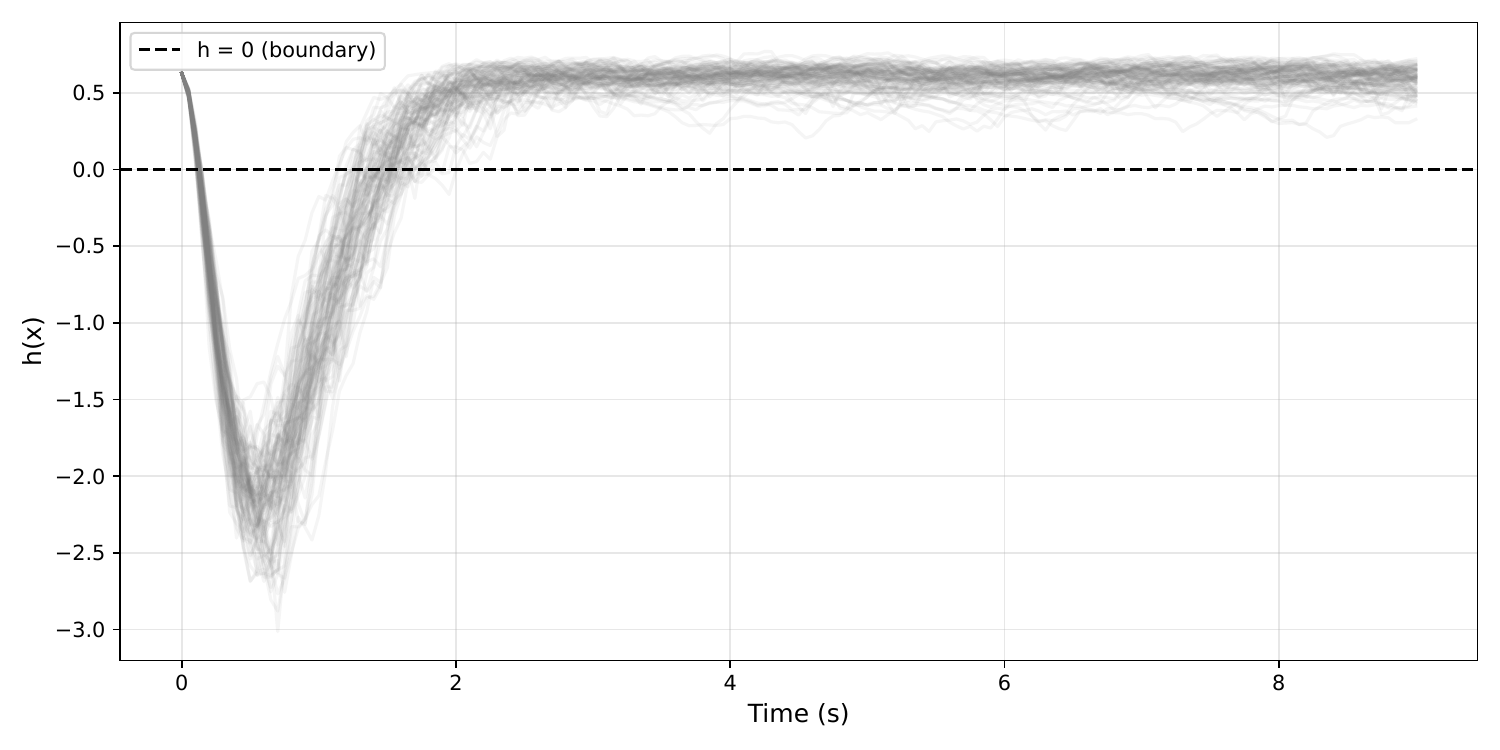}
        \caption{$h(x_k)$: Nominal}
        \label{fig:nominal-oval-h}
    \end{subfigure}
    \hfill
    \begin{subfigure}[b]{0.24\textwidth}
        \centering
        \includegraphics[width=\textwidth]{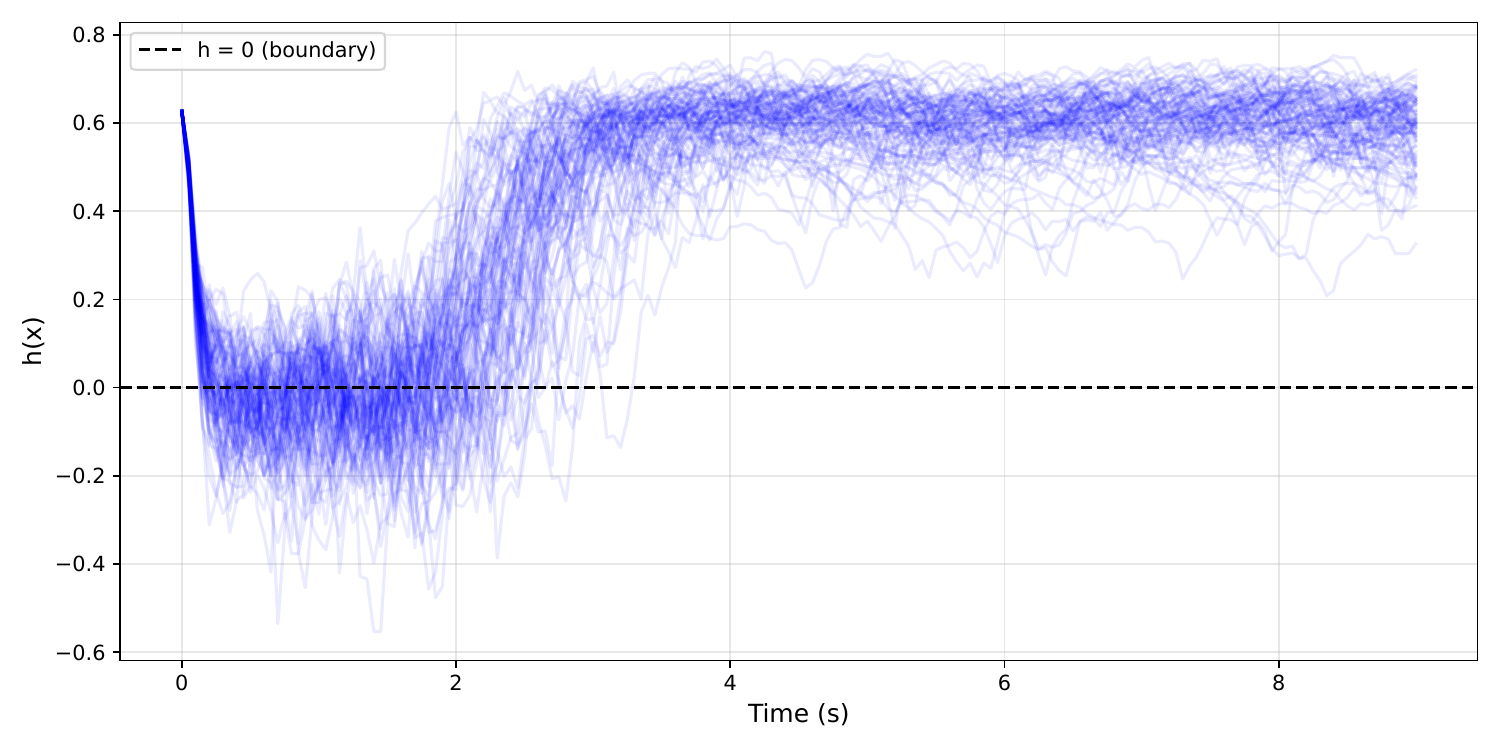}
        \caption{$h(x_k)$: Nominal CBF}
        \label{fig:nominal_cbf-oval-h}
    \end{subfigure}
    \hfill
    \begin{subfigure}[b]{0.24\textwidth}
        \centering
        \includegraphics[width=\textwidth]{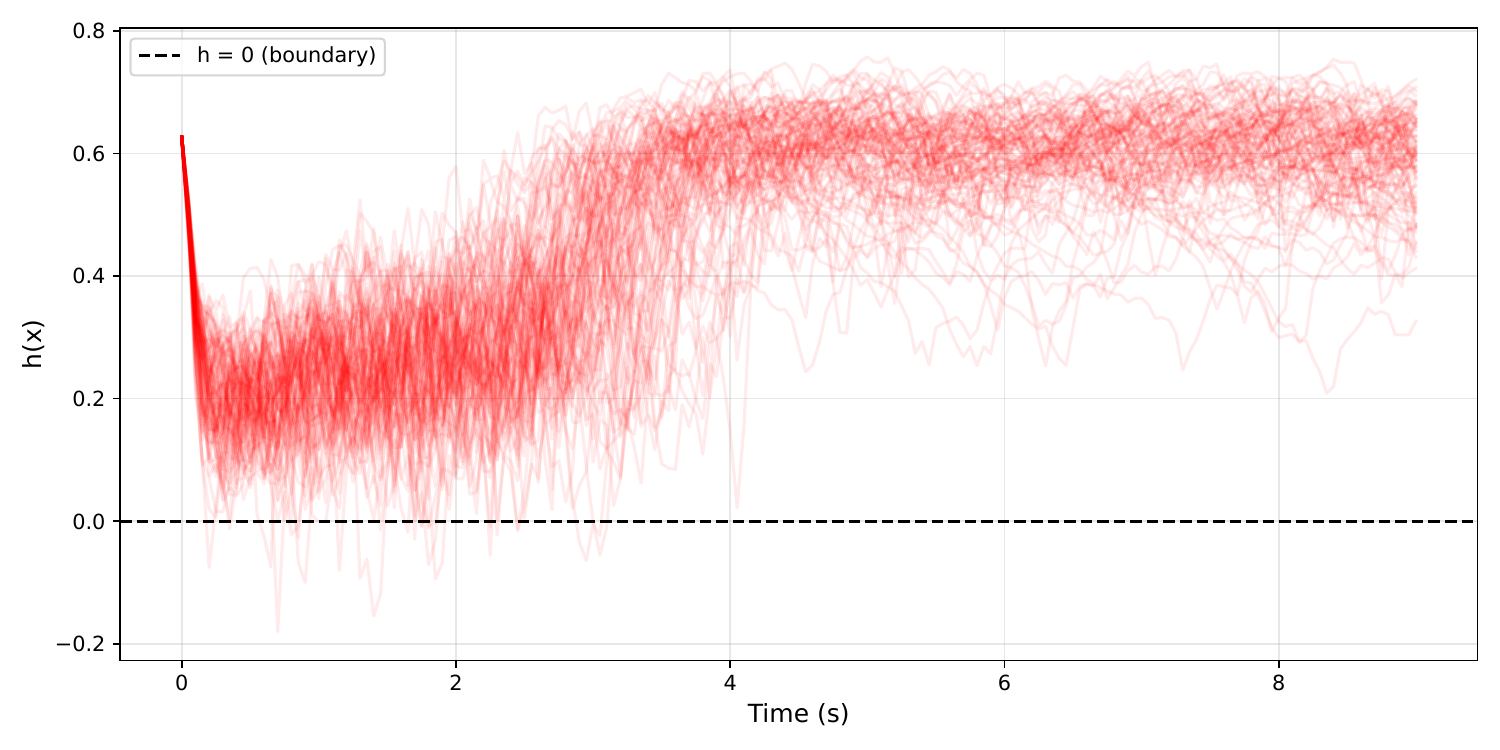}
        \caption{$h(x_k)$: Kalman-CVaR-CBF}
        \label{fig:scbf-oval-h}
    \end{subfigure}
    \hfill
    \begin{subfigure}[b]{0.24\textwidth}
        \centering
        \includegraphics[width=\textwidth]{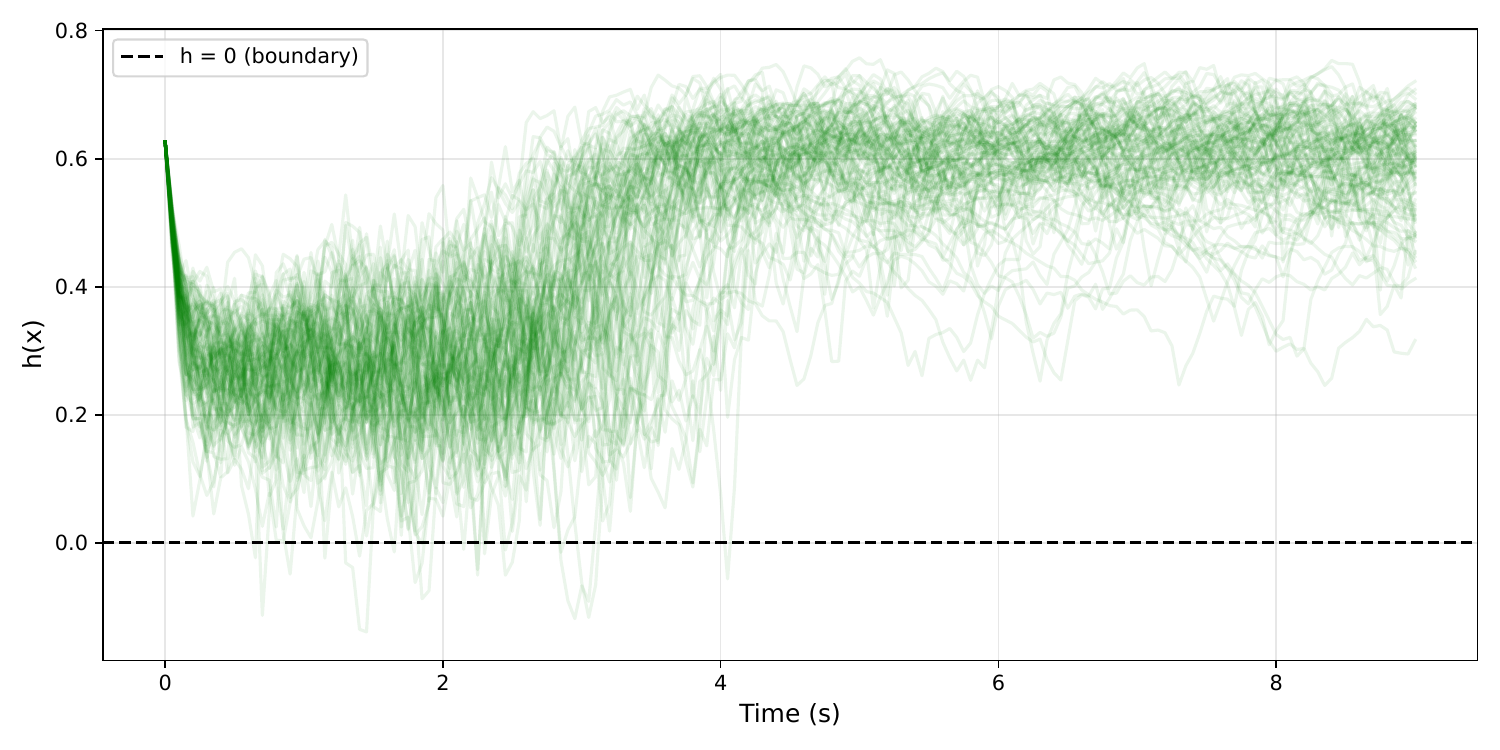}
        \caption{$h(x_k)$: Ours}
        \label{fig:cvar-oval-h}
    \end{subfigure}
    \caption{Closed-loop state trajectories and  evolution of the safety function $h(x_k)$ over 100 Monte Carlo trials under the ellipsoidal safety constraint. The empirical safety probabilities of the four controllers are $0$, $0$, $0.8$, and $0.8$, respectively.}
    \label{fig:oval_all}
\end{figure*}

\begin{figure}
    \centering
    \includegraphics[width=\linewidth]{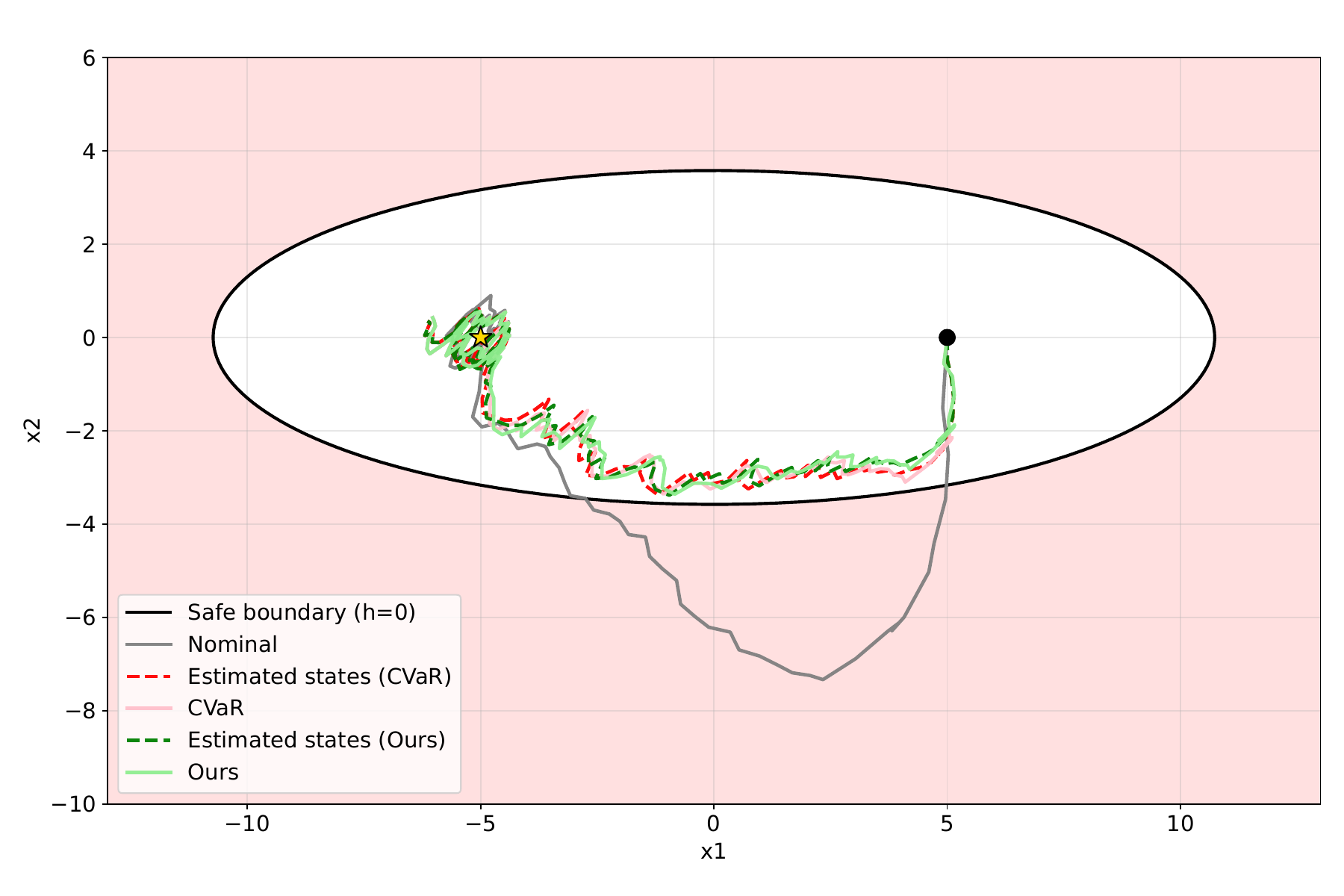}
    \caption{One realization of the closed-loop trajectories by the two methods under the ellipsoidal safety constraint. Both methods achieve the same empirical safety probability $\hat{P}_s=0.8$, with parameters $\alpha=0.52$ and $\epsilon=0.27$ for the method in \cite{kishida2024risk}, and $\alpha=0.52$ and $c_J'=0.38\times c_J^{\max}$ for our approach.}
    \label{fig:oval-comparison}
\end{figure} 

Then we set the parameters in our approach as $\alpha=0.7$ and $c_J'=0.115\times c_{J}^{{\max}}$,
where \vspace{-4pt}
\begin{flalign}\label{eq-cjmax}
    c_J^{\max}&:=(M-h_\gamma)(1-\alpha)+\frac{\lambda_{\max}}{2}\text{tr}(K_k\text{cov}[v_k]K_k^\top)\notag\\
    &~~~+\frac{\lambda_{\max}}{2}\text{tr}(K_kC_kP_kC_k^\top K_k^\top). \vspace{-4pt}
\end{flalign}
With this choice, our approach achieves the same empirical safety probability $\hat{P}_s=0.76$ as \cite{kishida2024risk}. This indicates that our approach can attain the same level of safety performance as \cite{kishida2024risk}. The resulting system trajectories, together with the corresponding Kalman filter-based state estimation trajectories obtained by the two approaches, are shown in Figure~\ref{fig:halfplane-comparison}.

Next, we further tune the parameters in both approaches so that they achieve the same empirical safety probability $\hat{P}_s$ among $\{0.4,0.6,0.76,0.8,0.9,0.99\}$ over 100 Monte Carlo realizations of stochastic disturbances and noises, and record the average runtime for each case. Specifically, for our approach, we tune $\alpha\!\in\!(0,1)$ and $k_J\!\in\![0,1]$ in \vspace{-3pt}
\[
c_J'=k_J\cdot c_J^{\max}, \vspace{-3pt}
\]
while for  \cite{kishida2024risk}, we tune $\alpha$ and $\epsilon$. The resulting parameters and runtimes are reported in Table~\ref{tab:halfplane}. We observe that, under the half-plane safety constraint, our approach achieves the same empirical safety probability as the method in \cite{kishida2024risk}, while maintaining a comparable level of computational time.

Furthermore, with $\alpha=0.7$, we also compare our approach with a \textbf{nominal CBF} controller defined by \vspace{-3pt}
\begin{equation}
    \hat{h}(A_k\hat{x}_k+B_ku_k)\geq\alpha \hat{h}(\hat{x}_k), \vspace{-3pt}
\end{equation}
with $\hat{h}$ defined in \eqref{eq-hhat}, as well as  the \textbf{nominal controller} without any safety constraint. For the method in \cite{kishida2024risk}, we use the same parameters $\alpha=0.7$ and $\epsilon=0.3$ as before.
We again perform 100 Monte Carlo trials to estimate the empirical safety probability of each controller. The resulting system trajectories, together with the evolution of the safety function $h(x_k)$, are shown in Figure~\ref{fig:halfplan_all}. The empirical safety probabilities of the four controllers are $0$, $0.05$, $0.76$, $0.76$, respectively. We observe that, compared with the nominal controller and the nominal CBF controller, both our approach and  \cite{kishida2024risk} significantly improve the safety performance.

\begin{table}[t]
  \centering
  \caption{Statistical Results for Half-Plane Constraint}
  \label{tab:halfplane}
  \setlength{\tabcolsep}{6pt}  
  \begin{tabular}{@{}cccccccc@{}}
    \toprule
    \multirow{2}{*}{Method} & \multirow{2}{*}{Metric} & \multicolumn{6}{c}{ Monte Carlo estimates $\hat{P}_s$} \\
    \cmidrule(l){3-8}
    & & 0.4 & 0.6 & 0.76 & 0.8 & 0.9 & 0.99 \\
    \midrule
    \multirow{3}{*}{\cite{kishida2024risk}} 
      & $\alpha$   & 0.55  & 0.7 & 0.7  & 0.7   & 0.75  & 0.8  \\
      & $\epsilon$ & 0.4 & 0.38 & 0.3 & 0.28 & 0.28 & 0.27 \\
      & time (ms)  & 9.00 & 7.43 & 10.21  & 7.87  & 7.71 & 7.64 \\
    \midrule
    \multirow{3}{*}{ours} 
      & $\alpha$   & 0.55  & 0.7 & 0.7  & 0.7   & 0.75  & 0.8  \\
      & $k_J$      & 0.077 & 0.102 & 0.115 & 0.125 & 0.14 & 0.17 \\
      & time (ms)  & 13.67  & 11.17 & 15.22  & 11.58  & 11.09 & 11.16 \\
    \bottomrule
  \end{tabular} 
\end{table}

\textbf{Comparisons under the ellipsoidal safety constraint.} We then consider the following ellipsoidal safety set  \vspace{-3pt}
\begin{equation}
    \mathcal{C}_2=\{x\in\RR^2: 0.8- x^\top\begin{bmatrix}
        1/144& 0\\
        0& 1/16
    \end{bmatrix}x\geq0 \}. \vspace{-3pt}\notag
\end{equation}
In this setting, we increase the covariances of stochastic disturbances to \vspace{-3pt}
\begin{equation}
    \cov[w_k]=\begin{bmatrix}
        0.03 & 0.03\\
        0.03 & 0.03
    \end{bmatrix},~\cov[v_k]=0.09.\notag \vspace{-3pt}
\end{equation}
The initial and target states are $x_0\!=\!\begin{bmatrix}
    5&0
\end{bmatrix}^\top$, $x_g\!=\!\begin{bmatrix}
    -5&0
\end{bmatrix}^\top$, respectively. As the nominal controller, we choose an LQR controller of the form\vspace{-3pt}
\begin{equation}
    u(x)=-K_{\text{lqr}}\left( x-x_g \right),~Q_{\text{lqr}}=\mathrm{diag}(1.0,0.5),~R_{\text{lqr}}=0.1.\vspace{-3pt}\notag
\end{equation}

Similarly to the previous set of simulations, we further tune the parameters in the two approaches so that they achieve the same empirical safety probability among $\{0.4,0.6,0.7,0.8,0.9,0.99\}$ over 100 Monte Carlo realizations of the stochastic disturbances and measurement noises. One realization of the resulting closed-loop trajectories generated by the two methods is shown in Figure~\ref{fig:oval-comparison}, and the corresponding parameters and average runtimes are reported in Table~\ref{tab:oval}.
We also compare the proposed approach with $\alpha=0.52$ and $k_J=0.38$ against the nominal controller, the nominal CBF controller with $\alpha=0.52$, and the method in \cite{kishida2024risk} with $\alpha=0.52$ and $\epsilon=0.27$. The resulting system trajectories, together with the evolution of the safety function $h(x_k)$, are shown in Figure~\ref{fig:oval_all}.

It is worth noting that, from Table~\ref{tab:oval}, we observe that, under the ellipsoidal safety constraint, our approach achieves the same empirical safety probability as \cite{kishida2024risk}, while requiring \textbf{substantially less computation time}.
This is mainly because the worst-case CVaR approach in \cite{kishida2024risk} involves a distributionally robust min-max structure. While this complexity can be simplified to a linear inequality for half-space safe sets, the ellipsoidal case leads to a quadratic-loss risk constraint together with auxiliary variables and additional inequalities in the online optimization. By contrast, our Jensen-based reformulation directly yields an explicit deterministic barrier constraint in terms of the control input.
As a result, the computational advantage of our approach becomes much more significant in the ellipsoidal case.

In summary, compared with the nominal controller and the nominal CBF controller, both our approach and \cite{kishida2024risk} can achieve substantially higher safety probabilities. Furthermore, compared with \cite{kishida2024risk}, our approach can attain a comparable level of safety performance with significantly improved computational efficiency, especially when the safety constraint becomes more complex. In fact, \cite{kishida2024risk} only considers half-plane and ellipsoidal safety constraints, whereas our approach can be applied to more general safety specifications, as long as the control barrier function $h$ satisfies conditions i)--iii) in Theorem~\ref{thm}. This demonstrates that our approach is not only computationally efficient, but also more broadly applicable for safe control of discrete-time stochastic systems with Kalman filtering-based state estimation.

\begin{table}[t]
  \centering
  \caption{Statistical Results for Ellipsoidal Constraint}
  \label{tab:oval}
  \setlength{\tabcolsep}{3pt}  
  \begin{tabular}{@{}cccccccc@{}}
    \toprule
    \multirow{2}{*}{Method} & \multirow{2}{*}{Metric} & \multicolumn{6}{c}{ Monte Carlo estimates $\hat{P}_s$} \\
    \cmidrule(l){3-8}
    & & 0.4 & 0.6 & 0.7 & 0.8 & 0.9 & 0.99 \\
    \midrule
    \multirow{3}{*}{\cite{kishida2024risk}} 
      & $\alpha$   & 0.3  & 0.45 & 0.5  & 0.52   & 0.6  & 0.8  \\
      & $\epsilon$ & 0.29 & 0.3 & 0.297 & 0.27 & 0.25 & 0.28 \\
      & time (ms)  & 2878.1 & 3253.9 & 2734.7  & 2839.1  & 3005.1 & 4981.2 \\
    \midrule
    \multirow{3}{*}{ours} 
      & $\alpha$   & 0.3  & 0.45 & 0.5  & 0.52  & 0.6  & 0.8  \\
      & $k_J$      & 0.258 & 0.31 & 0.353 & 0.38 & 0.44 & 0.7 \\
      & time (ms)  & 44.71  & 57.84 & 47.10  & 46.18  & 53.86 & 59.76 \\
    \bottomrule
  \end{tabular}
\end{table}

\subsection{Validation of Probabilistic Safety Guarantee}

In this subsection, we evaluate the theoretical probabilistic safety bound provided by Theorem~\ref{thm} against the empirical safety probability obtained from Monte Carlo simulations. Since our approach is developed for linear systems, we adopt the underactuated inverted-pendulum example in \cite{cosner2023robust} and linearize its dynamics around the upright equilibrium by approximating $\sin(\theta_k)$ with $\theta_k$ within the range $|\theta_k|\leq \pi/6$. Letting $x_t:=\begin{bmatrix}
    \theta_k&\dot{\theta}_k
\end{bmatrix}^\top$, the linearized system is given by \vspace{-3pt}
\begin{subequations}
    \begin{align}
        x_{k+1}&=\begin{bmatrix}
            1& \Delta t\\
            \Delta t & 1
        \end{bmatrix}x_k+\begin{bmatrix}
            0\\
            \Delta t
        \end{bmatrix}u_k+w_k\\
        y_k&=\begin{bmatrix}
            1 & 0
        \end{bmatrix}x_k+v_k \vspace{-3pt}
    \end{align}
\end{subequations}
where we use the same disturbance covariance as in \cite{cosner2023robust} and set the measurement noise covariance as \vspace{-3pt}
\begin{equation}
    \cov[w_k]=\mathrm{diag}(0.005^2,0.025^2),~\cov[v_k]=0.005^2. \vspace{-3pt}\notag
\end{equation}
As the nominal controller, we choose an LQR controller of the form \vspace{-3pt}
\begin{equation}
    u(x)=-K_{\text{lqr}}x,~Q_{\text{lqr}}=\mathrm{diag}(12,1),~R_{\text{lqr}}=0.2. \vspace{-3pt}\notag
\end{equation}
We define the same safety set as in \cite{cosner2023robust}, \vspace{-2pt}
\begin{equation}
    \mathcal{C}=\left\{ x\in\RR^2: 1-\frac{6^2}{\pi^2}x^\top\begin{bmatrix}
        1&3^{-\frac{1}{2}}\\
        3^{-\frac{1}{2}}&1
    \end{bmatrix}x\geq 0 \right\}. \vspace{-2pt}\notag
\end{equation}

We first reproduce the Jensen-based CBF controller in \cite{cosner2023robust} under perfect state information, using the same parameters $c_J=\frac{\lambda_{\max}}{2}\text{tr}(\text{cov}(w_k))$ and $\alpha=1-c_J$, and perform 100 Monte Carlo trials for different $x_0$. For each $x_0$, we compute both the theoretical safety bound given by Lemma~\ref{lem-expectationcbf} and the  Monte Carlo estimate of the safety probability.

For our output-feedback approach, we use the same value $\alpha=1-c_J$ and set $c_J'=0.2\times c_J^{\max}$, where $c_J^{\max}$ is defined in \eqref{eq-cjmax}. For Kalman filtering, we choose $\sigma=0.05$. We then again perform 100 Monte Carlo trials for each initial state and compute the corresponding theoretical safety bound from Theorem~\ref{thm} together with the empirical safety estimate.
Figure~\ref{fig:oval_heat_all} presents the closed-loop trajectories starting from the initial state $x_0=0$, together with the theoretical safety guarantees and the Monte Carlo estimates of the safety probability from different initial states $x_0$. 

From Figure~\ref{fig:oval_heat_all}, we observe that, compared with the state-feedback CBF controller in \cite{cosner2023robust} under perfect state information, our approach yields a slightly lower theoretical safety probability due to the additional estimation error introduced by Kalman filtering and its  probability guarantee $1-\sigma$. Nevertheless, the Monte Carlo estimates show that both the state-feedback controller \cite{cosner2023robust} and our approach provide reliable characterizations of the actual safety performance, in the sense that the empirical safety probabilities are consistently higher than the corresponding theoretical safety bounds.

\begin{figure}[htbp]
      \centering
      \begin{subfigure}[b]{0.49\linewidth}
          \centering
          \includegraphics[width=0.9\textwidth]{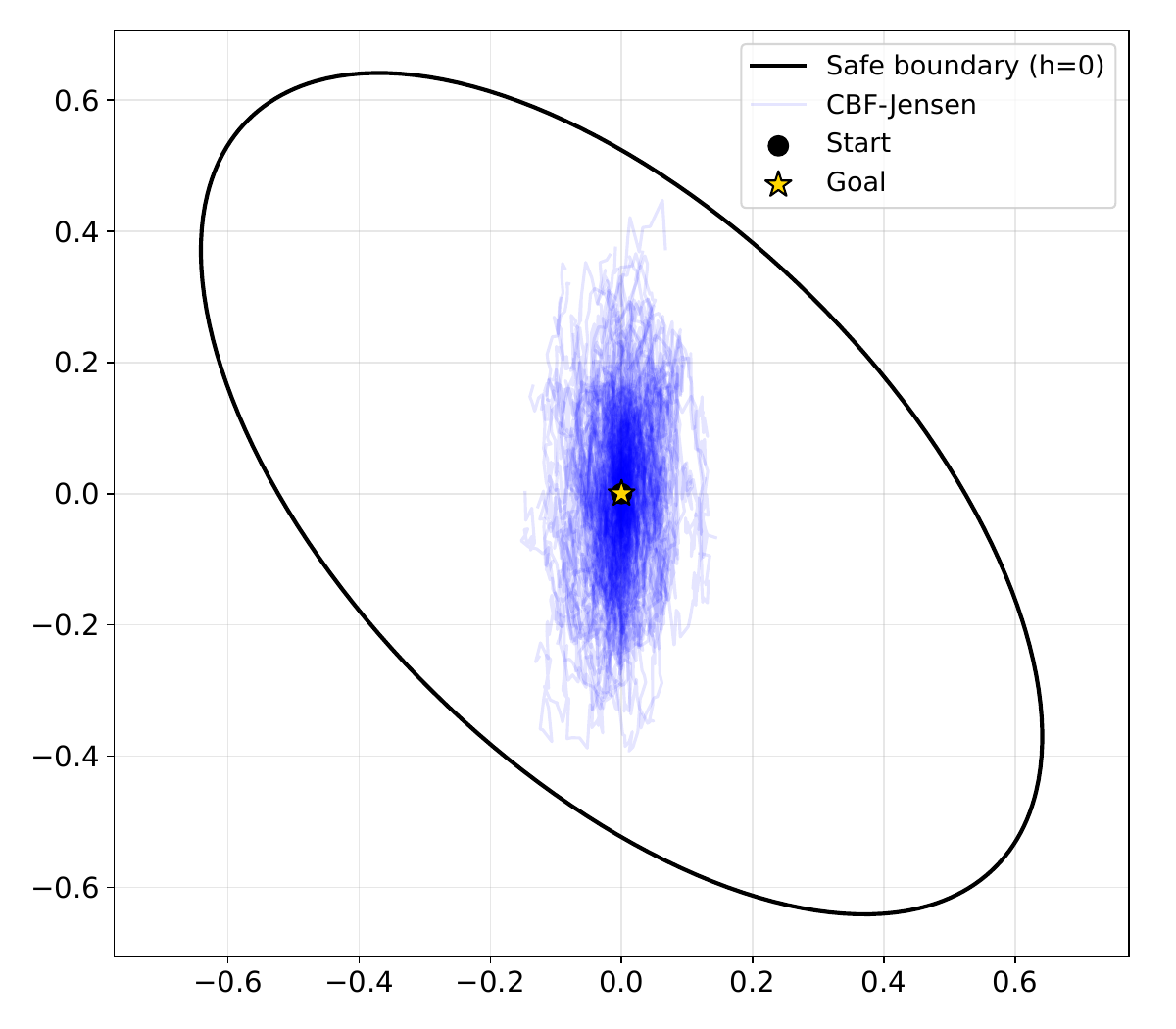}
          \caption{State-feedback trajectories}
          \label{fig:jensen-traj}
      \end{subfigure}
      \hfill
      \begin{subfigure}[b]{0.49\linewidth}
          \centering
          \includegraphics[width=0.9\textwidth]{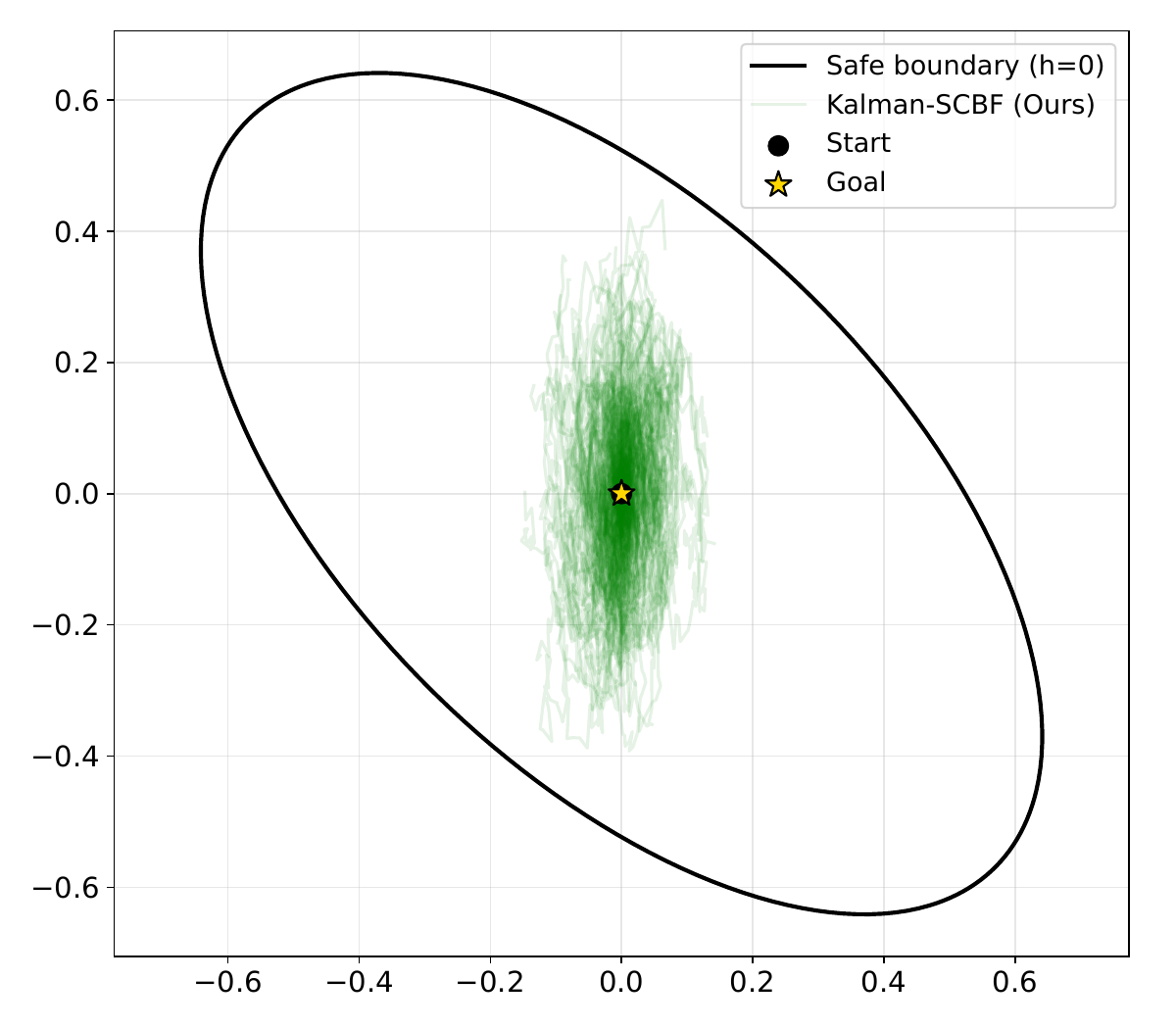}
          \caption{Output-feedback trajectories}
          \label{fig:scbf-traj}
      \end{subfigure}


      \begin{subfigure}[b]{0.49\linewidth}
          \centering
          \includegraphics[width=\textwidth]{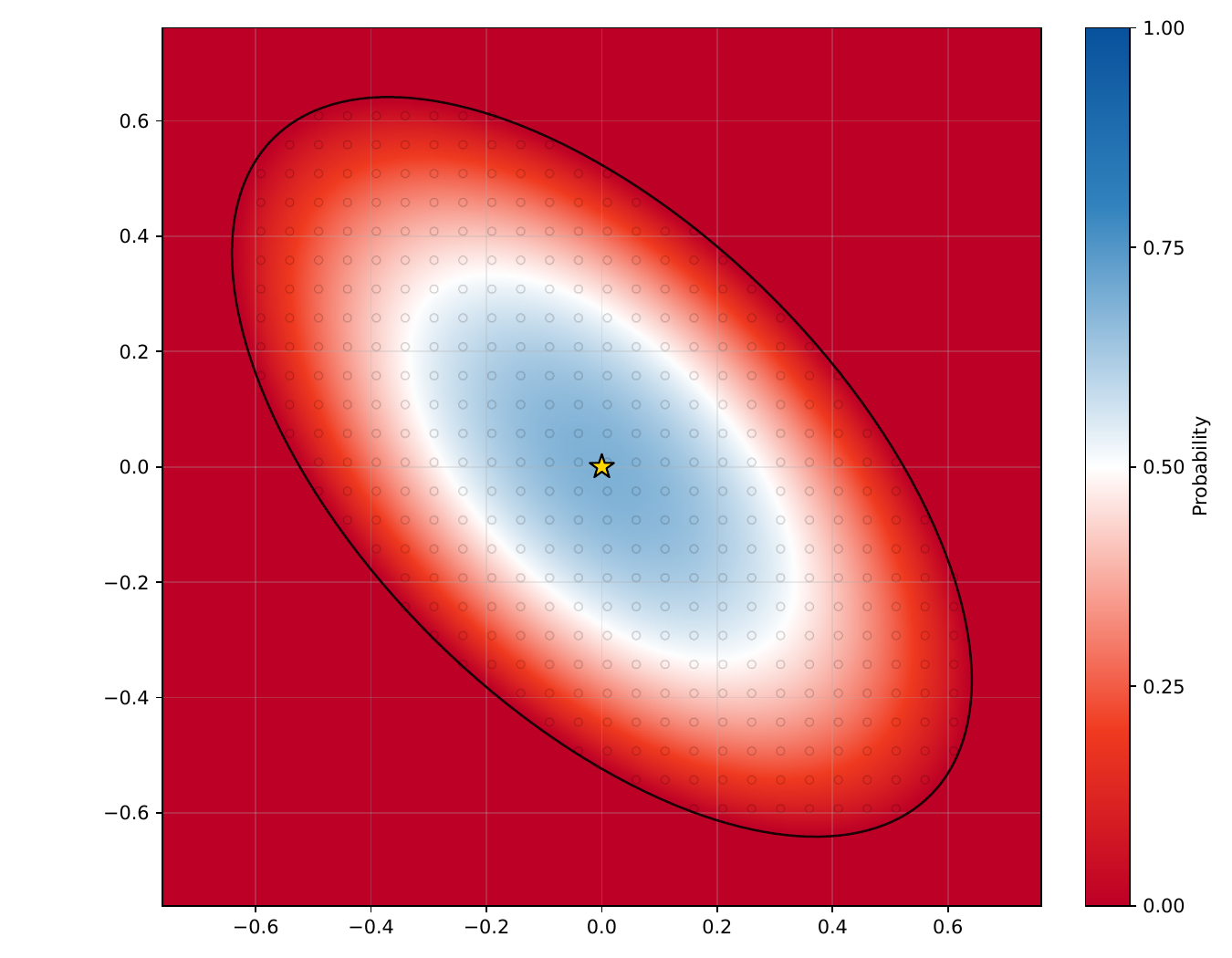}
          \caption{State-feedback bound}
          \label{fig:jensen-heat}
      \end{subfigure}
      \hfill
      \begin{subfigure}[b]{0.49\linewidth}
          \centering
          \includegraphics[width=\textwidth]{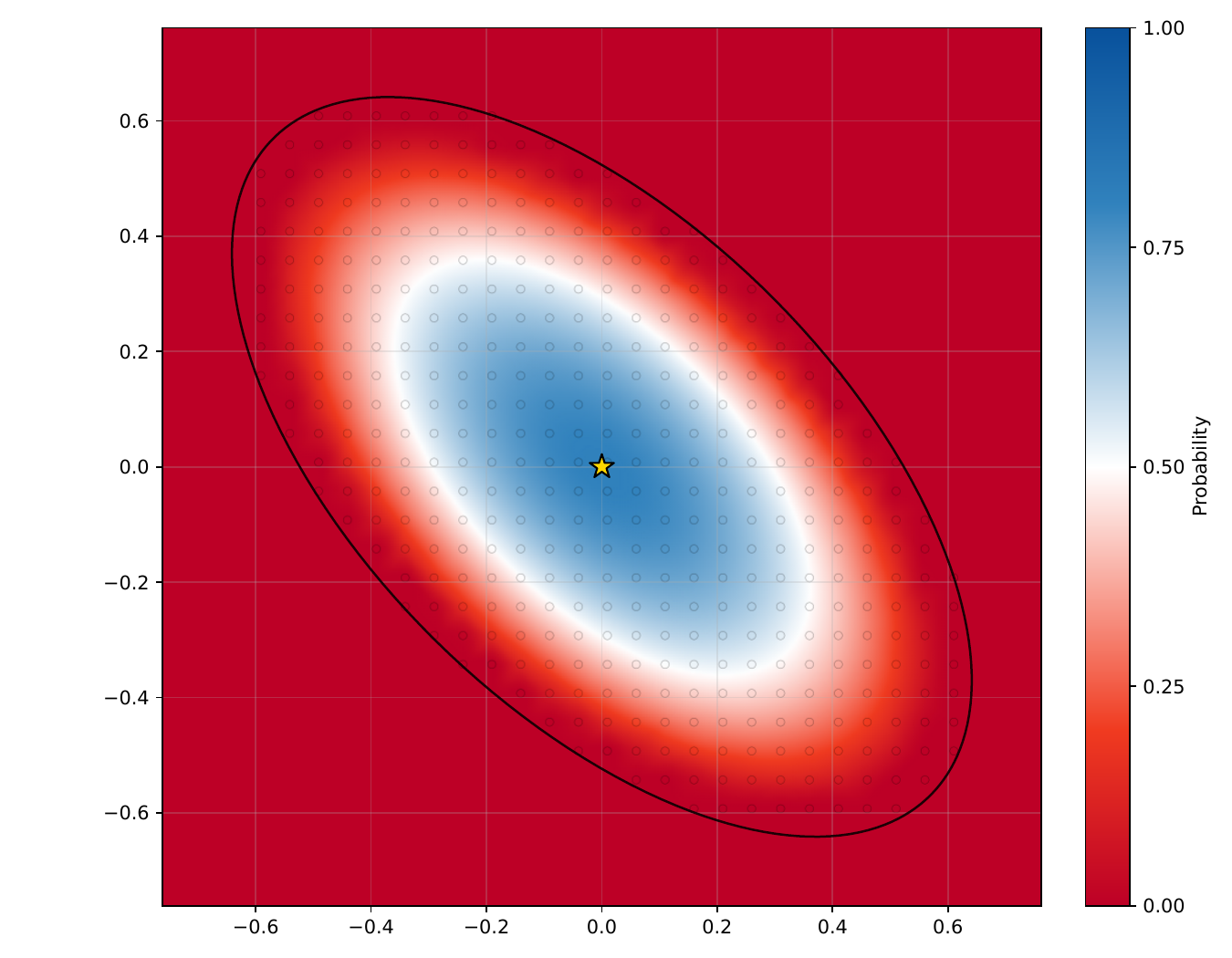}
          \caption{Output-feedback bound}
          \label{fig:scbf-heat}
      \end{subfigure}


      \begin{subfigure}[b]{0.49\linewidth}
          \centering
          \includegraphics[width=\textwidth]{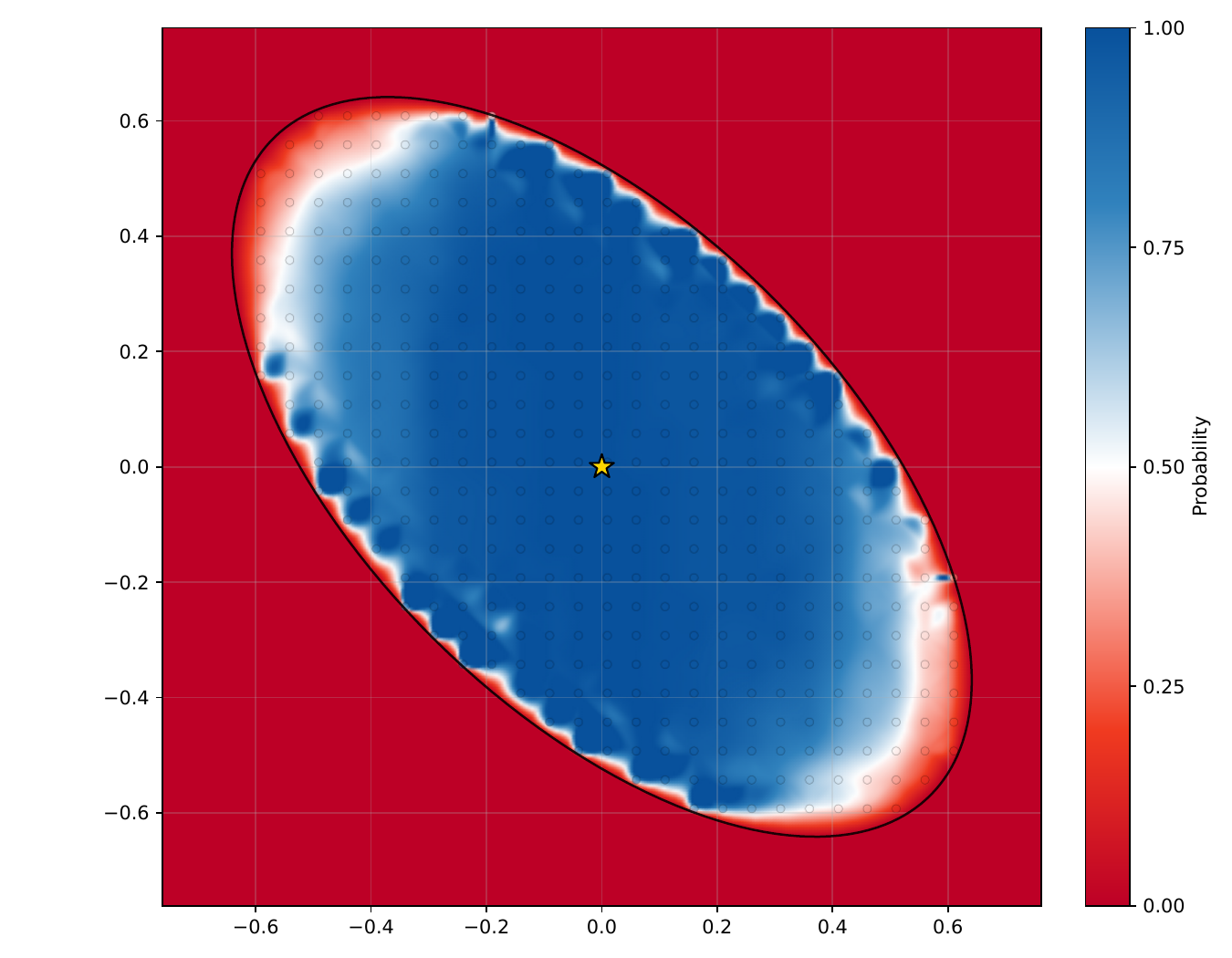}
          \caption{State-feedback estimate}
          \label{fig:jensen-real-heat}
      \end{subfigure}
      \hfill
      \begin{subfigure}[b]{0.49\linewidth}
          \centering
          \includegraphics[width=\textwidth]{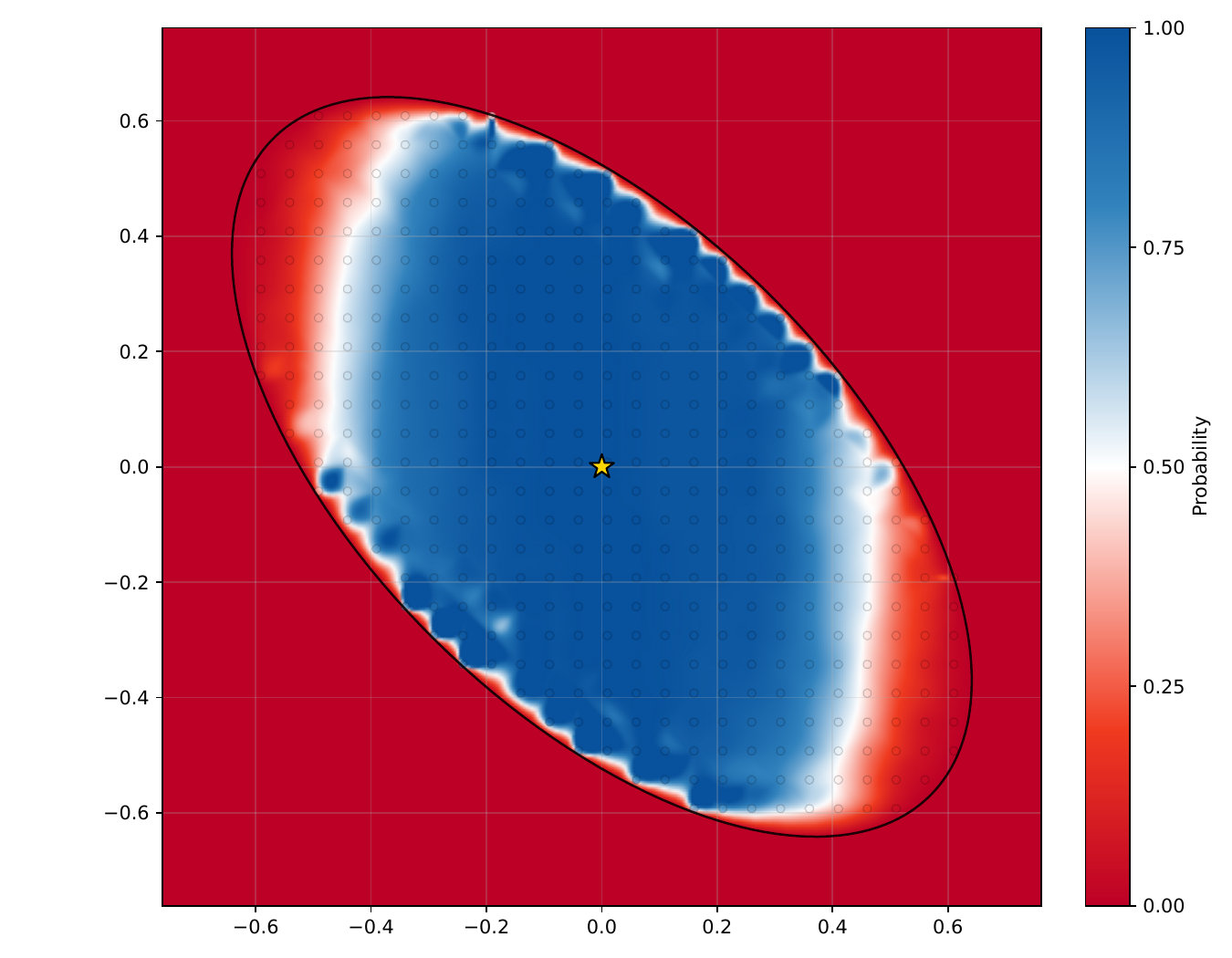}
          \caption{Output-feedback estimate}
          \label{fig:scbf-real-heat}
      \end{subfigure}

      \caption{Comparison between the state-feedback controller and our approach. The first row shows system trajectories over 100 Monte Carlo trials starting from $x_0=0$. The second row shows the theoretical safety probability bounds over different initial states, and the third row shows the corresponding Monte Carlo estimates of the safety probability. }
      \label{fig:oval_heat_all}
\end{figure}

\section{Conclusion}\label{sec:conclusion}

In this paper, we study the safe control for discrete-time stochastic systems with incomplete state information. By integrating Kalman filtering into the stochastic control barrier function framework, we construct an estimate-based safety function and develop a Jensen-based deterministic reformulation that yields explicit finite-horizon probabilistic safety guarantees for the resulting output-feedback controller. Simulation results demonstrated that our approach achieves reliable safety performance, and can attain comparable empirical safety probabilities with improved computational efficiency, especially for more complex safety constraints. 
Future work will consider extensions of the proposed framework to model predictive control schemes and its application to more complex robotic systems.

\bibliographystyle{IEEEtran}
\bibliography{myref}

\end{document}